\documentclass[journal,twoside,web]{ieeecolor}

\usepackage{generic}
\usepackage{cite}
\usepackage{amsmath,amssymb,amsfonts}
\usepackage{graphicx}
\usepackage{textcomp}

\usepackage{amsthm} 
\usepackage{newtxmath} 
\usepackage[euler]{textgreek} 
\let\labelindent\relax 
\usepackage[inline]{enumitem} 
\usepackage{stmaryrd} 
\usepackage{hhline} 
\usepackage{colortbl} 
\usepackage{cancel} 
\usepackage{setspace} 

\usepackage{siunitx} 
\DeclareSIUnit{\byte}{B}

\definecolor{lightgray}{gray}{0.83}

\makeatletter
\let\NAT@parse\undefined
\makeatother
\usepackage{hyperref}
\hypersetup{ %
  colorlinks=true, %
  pdfstartview={FitH}, %
  citecolor=black, 
  linkcolor=black, %
  urlcolor=black 
 }

\newcommand{\gvec}{\mathbf{g}}

\newcommand{\rvec}{\mathbf{r}}

\newcommand{\uvec}{\mathbf{u}}

\newcommand{\yvec}{\mathbf{y}}

\newcommand{\partialfrac}[2]{\frac{\partial #1}{\partial #2}}
\newcommand{\jacobian}[2]{\partialfrac{#1}{#2}}
\newcommand{\norm}[1]{\left\lVert#1\right\rVert}

\DeclareMathOperator*{\RMS}{RMS}

\newcommand{\real}{\mathbb{R}}

\newcommand{\integer}{\mathbb{Z}}

\usepackage[normalem]{ulem}
\definecolor{revyellow}{RGB}{150,150,0}
\definecolor{revblue}{RGB}{0,0,220}
\definecolor{revorange}{RGB}{200,100,0}
\definecolor{revgreen}{RGB}{0,135,0}
\definecolor{revpurp}{RGB}{150,0,150}

\theoremstyle{definition}
\newtheorem{definition}{Definition}
\theoremstyle{plain}
\newtheorem{theorem}{Theorem}
\newtheorem{corollary}{Corollary}[theorem]
\newtheorem{lemma}{Lemma}
\theoremstyle{remark}
\newtheorem{remark}{Remark}

\newcommand{\eqlabel}[1]{\addtocounter{equation}{-1}\refstepcounter{equation}\label{#1}}

\newcounter{subeq}

\newcommand{\suspend}[1]{
\newcounter{#1}\setcounter{#1}{\value{enumi}}
}
\newcommand{\resume}[1]{
\setcounter{enumi}{\value{#1}}
}

\newcommand{\sysDef}{(\ref{eq:sysDefPWA})}
\newcommand{\baseAssumptions}{\ref{C:reachable}-\ref{C:muc}}
\newcommand{\controlModel}{(\ref{eq:monox})-(\ref{eq:monoy})}

\makeatletter
\renewcommand*\env@matrix[1][\arraystretch]{%
  \edef\arraystretch{#1}%
  \hskip -\arraycolsep
  \let\@ifnextchar\new@ifnextchar
  \array{*\c@MaxMatrixCols c}}
\makeatother

\newcommand{\nonmin}{non-minimum}

\newcommand{\ILILC}{ILILC}

\newcommand{\PWA}{PWA}

\newcommand{\NMP}{NMP}

\newcommand{\offsetVec}{\beta}
\newcommand{\state}{x}
\newcommand{\uMIMO}{u}
\newcommand{\yMIMO}{y}

\newcommand{\region}{Q}
\newcommand{\regionQuant}{{|\region|}}

\newcommand{\learnMat}{L}
\newcommand{\tdx}{\ell}

\newcommand{\idx}{\kappa} 
\newcommand{\xmod}{\hat{x}}

\newcommand{\ymod}{\hat{y}}
\newcommand{\yvecmod}{\hat{\yvec}}

\newcommand{\gmod}{\hat{\gvec}}

\newcommand{\force}{c}

\newcommand{\stab}{{\mathscr{s}}}
\newcommand{\unstab}{{\mathscr{u}}}

\newcommand{\kappadx}{_\kappa}
\newcommand{\kdxplus}[1]{_{k+#1}}
\newcommand{\kdxminus}[1]{_{k-#1}}
\newcommand{\kdx}{_k}
\newcommand{\dmu}{\mu_g} 
\newcommand{\Qset}{\region}
\newcommand{\Asum}{\mathbf{A}}
\newcommand{\Bsum}{\mathbf{B}}
\newcommand{\Csum}{\mathbf{C}}
\newcommand{\Dsum}{\mathbf{D}}
\newcommand{\Fsum}{\mathbf{F}}
\newcommand{\Gsum}{\mathbf{G}}
\newcommand{\Msum}{\mathbf{M}}
\newcommand{\sumqQ}{\sum_{q=1}^{\regionQuant}}
\newcommand{\locvec}{\delta}
\newcommand{\sigvec}{\delta^*}
\newcommand{\sigvecset}{\Delta^*}
\newcommand{\offsetvec}{\offsetVec}
\newcommand{\locvecmod}{\hat{\locvec}}
\newcommand{\sigvecsetmod}{\hat{\Delta}^*}
\newcommand{\inv}{\overline}
\newcommand{\prevC}{\mathcal{C}}
\newcommand{\prevD}{\mathcal{D}}
\newcommand{\prevG}{\mathcal{G}}
\newcommand{\anticaus}{\Psi}

\newcommand{\decoup}{\tilde}
\newcommand{\statedc}{\decoup{\state}}
\newcommand{\Asumdc}{\decoup{\Asum}}
\newcommand{\Bsumdc}{\decoup{\Bsum}}
\newcommand{\Fsumdc}{\decoup{\Fsum}}
\newcommand{\Pdc}{\decoup{P}}
\newcommand{\extracts}{\mathscr{I}^\stab}
\newcommand{\extractu}{\mathscr{I}^\unstab}
\newcommand{\refr}{r}
\newcommand{\controltotal}{\force}
\newcommand{\controlfb}{y}
\newcommand{\Acm}{\hat{A}} 
\newcommand{\Bcm}{\hat{B}}
\newcommand{\Ccm}{\hat{C}}
\newcommand{\Dcm}{\hat{D}}
\newcommand{\xcm}{\hat{x}}
\newcommand{\statemod}{\xmod}
\newcommand{\Asummod}{\hat{\Asum}}
\newcommand{\Bsummod}{\hat{\Bsum}}
\newcommand{\Fsummod}{\hat{\Fsum}}
\newcommand{\Csummod}{\hat{\Csum}}
\newcommand{\Dsummod}{\hat{\Dsum}}
\newcommand{\Gsummod}{\hat{\Gsum}}
\newcommand{\ldx}{\tdx}
\newcommand{\uLift}{\uvec}
\newcommand{\rLift}{\rvec}
\newcommand{\yLift}{\yvec}
\newcommand{\yLiftmod}{\hat{\yLift}}
\newcommand{\gLift}{\gvec}
\newcommand{\ginvLift}{\hat{\gLift}^{-1}}
\newcommand{\AcmC}{\hat{\mathbf{\Asum}}\kdx^C}
\newcommand{\BcmC}{\hat{\mathbf{\Bsum}}\kdx^C}
\newcommand{\CcmC}{\hat{\mathbf{\Csum}}\kdx^C}
\newcommand{\DcmC}{\hat{\mathbf{\Dsum}}\kdx^C}

\def\lownabla{\rule{0pt}{1.75ex}\nabla}
\def\lowP{\rule{0pt}{1.75ex}P}
\newcommand{\gradILCgain}{\gamma_{\lownabla}}
\newcommand{\PILCgain}{\gamma_{\lowP}}

\def\BibTeX{{\rm B\kern-.05em{\sc i\kern-.025em b}\kern-.08em

    T\kern-.1667em\lower.7ex\hbox{E}\kern-.125emX}}

\begin{document}%
\title{%
Stable Inversion of Piecewise Affine Systems%
\\[0.2ex]
\LARGE
with Application to Feedforward and Iterative Learning Control%
}
\author{
Isaac A. Spiegel, %
Nard Strijbosch, %
Robin de Rozario, %
Tom Oomen, %
and Kira Barton
\thanks{%
This work is supported by the National Science Foundation (CAREER Award \#1351469), the U.S. Department of Commerce, National Institute of Standards and Technology (Award 70NANB20H137), and the Netherlands Organization for Scientific Research (research programme VIDI project 15698)
\emph{(Corresponding author: Isaac A. Spiegel)}%
}%
\thanks{%
Isaac A. Spiegel is with the Department of Mechanical Engineering and Kira Barton is with the Departments of Mechanical Engineering and Robotics, University of Michigan, Ann Arbor, MI 48109 USA (e-mail: ispiegel@umich.edu, bartonkl@umich.edu).%
}%
\thanks{%
Nard Strijbosch and Robin de Rozario are with the
Department of Mechanical Engineering, Eindhoven University of Technology, 5612 AZ Eindhoven, the Netherlands
(e-mail: n.strijbosch@gmail.com, robinderozario@gmail.com).
}%
\thanks{%
Tom Oomen is with the Department of Mechanical Engineering, Eindhoven University of Technology, 5612 AZ Eindhoven, the Netherlands, and the Delft Center for Systems and Control, Delft University of Technology, 2628 CD Delft, the Netherlands (e-mail: t.a.e.oomen@tue.nl).
}%
\thanks{%
\copyright2024 IEEE. Personal use of this material is permitted. Permission from IEEE must be obtained for all other uses, 
including reprinting/republishing this material for advertising or promotional purposes, creating new collective works, for resale or redistribution to servers or lists, or reuse of any copyrighted component of this work in other works. This document is an accepted version. Full citation to published version:
I. A. Spiegel, N. Strijbosch, R. d. Rozario, T. Oomen and K. Barton, "Stable Inversion of Piecewise Affine Systems with Application to Feedforward and Iterative Learning Control," in IEEE Transactions on Automatic Control, doi: \href{https://doi.org/10.1109/TAC.2024.3382340}{10.1109/TAC.2024.3382340}
}
}%

\maketitle

\begin{abstract}
Model inversion is a fundamental technique in feedforward control. Unstable inverse models present a challenge in that
useful
feedforward control trajectories cannot be generated by directly propagating them.
Stable inversion is a process for generating useful trajectories from unstable inverses by handling their stable and unstable modes separately.
Piecewise affine (PWA) systems are a popular framework for modeling complicated 
dynamics.
The primary contributions of this article are 
closed-form inverse formulas for a general class of PWA models, and stable inversion methods for these models.
Both contributions leverage closed-form 
model representations 
to prove sufficient conditions for 
solution existence and uniqueness,
and to develop 
solution
computation methods.
The result is implementable feedforward control synthesis from PWA models with either stable or unstable inverses.
In practice, feedforward control alone may yield substantial tracking errors due to mismatch between the known system model and the unknowable complete system physics.
Iterative learning control (ILC) is a technique for achieving robustness to model error in feedforward control.
To
demonstrate the primary contributions' validity and utility,
this article 
also integrates
PWA stable inversion with ILC in simulations 
based on
a physical printhead positioning system.
\end{abstract}


\newpage
\section{Introduction}
This article demonstrates for the first time the synthesis of feedforward control and iterative learning control (ILC) from piecewise affine (\PWA) models with unstable 
inverses.
To achieve this end result, 
several fundamental theoretical contributions are made on model inversion.
This section first introduces 
\PWA{} systems, feedforward control, stable inversion, and ILC, 
then gives an overview on the state of the art and its limitations at the intersection of these topics, and finally concretely itemizes the novel contributions of this article.

\PWA{}
systems are a class of 
hybrid
dynamical
system
in which a state space is partitioned into polytopic regions, each of which may be associated with different affine system dynamics \cite{Heemels2001,Sontag1981}\footnote{
Less formally, 
many \PWA{} systems can be conceived simply as traditional state space systems but with piecewise defined equations.
}.
\PWA{} systems have risen to great popularity 
because they ease the rigorous mathematical modeling of systems with behaviors 
too complex to be captured reasonably by traditional frameworks.
Examples include current transformers \cite{Ferrari-Trecate2003}, one-sided spring supports \cite{Bonsel2004}, and additive manufacturing systems \cite{Spiegel2020}.
The mathematical rigor provided by the \PWA{} framework facilitates analysis and control theory development for these systems.

Given a system model mapping from the input to the output, basic feedforward control is 
the
inputting 
of 
a desired output (i.e. reference) into the model's inverse to generate an input 
that is
expected to yield the reference when applied to the original system.
This process
fails when the inverse model is unstable.
Stable models with unstable inverses are frequently called ``\nonmin{} phase'' (\NMP) 
models\footnote{
Considering that the inverse of a single-input-single-output (SISO) system is the reciprocal of its transfer function, an example of an \NMP{} system is a linear discrete-time system with all poles inside the unit circle and at least one zero outside the unit circle.
}.
Study of \NMP{} systems is valued because 
many practical systems feature \NMP{} dynamics, including motor-tachometer assemblies \cite{Awtar2004}, power converters \cite{Tallfors2004,Escobar1999}, and inkjet printheads \cite{Ezzeldin2011}.

Stable inversion is a technique for achieving feedforward control synthesis from \NMP{} models \cite{Zundert2016}, including nonlinear ones \cite{Spiegel2021}.
The stable inversion process can be broken down into four parts.
First, decouple the stable and unstable modes of the unstable inverse 
model\footnote{
For linear state space models, this is done via a change of basis that yields a representation of the system with the state matrix in Jordan form.
}.
Second, propagate the stable modes forward in time from an initial condition in the usual manner.
Third, propagate the unstable modes backwards in time from a prescribed terminal condition. This yields a bounded state trajectory because 
a scalar difference equation that is unstable forwards in time is stable backwards in time\footnote{
As an illustrative example, consider
that 
$x\kdxplus{1}=1.25x\kdx$
is unstable but 
${x\kdx = \frac{1}{1.25}x\kdxplus{1}}$
is a stable propagation from $x\kdxplus{1}$ to $x\kdx$.
}.
Finally, reassemble the stable and unstable mode trajectories into a complete, bounded state trajectory of the inverse model and extract the resultant feedforward input trajectory.

ILC is the application of a learning function to the error trajectory
from a previous attempt at an output reference tracking task to generate the feedforward 
input trajectory for the next attempt. In practice, ILC is typically applied to systems performing a task repetitively in order to compensate for unmodeled dynamics and other repetitive disturbances 
\cite{Bristow2006}\footnote{
ILC can also be used as an alternative to closed-form model inversion 
for the generation of a particular approximate model inverse 
trajectory
by generating the error signal with the model to be inverted \cite{Markusson2001}.
},
to which the above feedforward techniques are sensitive.

The state of the art 
in control with \PWA{} models
mostly focuses on feedback control, with
examples including
stabilizing state feedback control \cite{Mignone2000,Hejri2021}
and model reference adaptive control \cite{DiBernardo2013}.
Being feedback-driven means these
controllers
require real-time sensing and necessarily feature transient errors due to feedback control's reactive nature.
When processes demand low-error 
reference tracking or when real-time sensing is prohibitively challenging,
feedforward control is often desirable.
However, 
feedforward control
has not been
as
thoroughly addressed in the \PWA{} 
literature%
\footnote{
In addition to preventing the control of systems with specific needs for feedforward control, this gap  inhibits the 
use
of existing feedback control theory that requires feedforward 
components. 
For example,
\cite{VandeWouw2008} presents 
an output reference tracking scheme
for 
\PWA{} systems
that uses
both feedback and feedforward control elements, but does not present a method to compute the feedforward 
signal. This limits
application to the synchronization of a 
follower
 system to a
leader
with the feedforward input known in advance.
}.

Recent work has achieved indirect feedforward control of some piecewise defined systems via ILC \cite{Spiegel2019,Strijbosch2020}, but the proposed ILC synthesis 
technique generally
fails when used with \NMP{} models \cite{Spiegel2021}.
Specifically, \NMP{} dynamics 
create a large condition number for a matrix that must be inverted in the learning 
function,
causing large inversion error.
To enable ILC for nonlinear \NMP{} systems
that are \emph{not} piecewise defined,
\cite{Spiegel2021}
produces a new ILC framework---Invert-Linearize ILC (\ILILC)---and incorporates it with stable inversion.
The combination of \ILILC{} and stable inversion is thus promising for the ILC and feedforward control of \PWA{} systems, including those with NMP models. 
However, at present there are two major issues that hamper realization of such control.
First, there does not exist any published theory on stable inversion for \PWA{} systems (or hybrid systems of any kind).
Second, there does not exist sufficient literature on 
the 
conventional
closed-form
\PWA{} system inversion 
prerequisite for stable inversion.
This is the core limitation preventing both \emph{direct} (as opposed to ILC-based) synthesis of feedforward control from all \PWA{} models and \emph{any} synthesis from \NMP{} \PWA{} models.

The most relevant prior art on 
\PWA{} 
model
inversion is 
the
foundational work on 
piecewise linear systems 
\cite{Sontag1981}, which
proposes that these systems are potentially 
invertible-with-delay, and explains the signal time shifting necessary to accommodate this delay.
This serves as a beginning for the concept of the relative degree of a \PWA{}
system\footnote{
The relative degree of a discrete-time system is the number of time steps it takes for an input value to influence the output value; for linear systems this is equal to the number of poles minus the number of zeros.
}.
Further work on the abstract algebra of piecewise linear systems states
that if a piecewise linear system has an inverse, it has a piecewise linear inverse \cite{Sontag1982}\footnote{
Also provided are the facts 
that it is decidable whether two sets are isomorphic under piecewise linear functions, and that 
the class of
piecewise linear functions 
is
closed under composition.
}.
There are several key limitations in this prior art on \PWA{} model inversion.
First, 
the relative-degree-related concepts
require further development 
for \PWA{} systems
to account for the facts that 
a \PWA{} system's apparent relative degree
may change during switching between component models and that \PWA{} systems may have multiple inverses. 
Second, 
\cite{Sontag1981}
leaves the derivation of the inverse dynamics other than the delay open as a ``nontrivial part'' of the inversion process. Third, for reference tracking it is desirable to invert a system \emph{without} 
delay\footnote{
Delay-free inversion frequently results in anticausal inverses, but this is no problem for typical feedforward control scenarios where the entire reference is known in advance. An anticausal inverse is one in which future reference values, and not present or past reference values, are required to compute a current input.
}.
Fourth, 
beyond the suggestion of a direction for future work in abstract algebra, 
no information is given on the uniqueness of the inverse, 
which is required for the inverse formula to be explicit rather than implicit.
Finally, this prior art does not address time-varying systems, which must be considered if one is to perform feedforward control of a system also subject to feedback control with a time-varying reference.

In short, 
the lack of general \PWA{} system inversion theory stymies the output reference tracking control of a broad range of practical systems, especially those with NMP models, for which ILC-based workarounds are currently inapplicable. The present work fills this gap by contributing
\begin{enumerate}[label=(C\arabic*),leftmargin=*]
    \item
    \label{contribution:inversion}
    formal inversion formulas for a class of \PWA{} systems,
    \item
    \label{contribution:uniqueness}
    sufficient conditions for inverse uniqueness and the corresponding explicit inverse representations,
    \item
    \label{contribution:SI}
    stable inversion 
    theory and implementation
    methods for \PWA{} models with NMP components, 
    and
    \item
    \label{contribution:sim}
    a 
    physically motivated
    simulation validation of this theory via the \ILILC{} of an NMP \PWA{} system.
\end{enumerate}
The main difficulties 
of
these contributions are
the fact
that piecewise definition has to date inhibited algebraic derivation of 
inverses,
and
the
determination of
how switching analytically influences the 
uniqueness of a \PWA{} model inverse and the ability to decouple its stable and unstable modes.

Section \ref{sec:sysdef} formally defines the class of \PWA{} systems treated in this work and a concept of relative degree for 
them.
Section \ref{sec:exactinverse} 
derives
inverse \PWA{} system formulas and sufficient conditions for inverse uniqueness, 
\ref{contribution:inversion}-\ref{contribution:uniqueness}, via analysis of the system's closed-form representation.
Section \ref{sec:stabinv} presents the stable inversion theory for \PWA{} systems,
\ref{contribution:SI}, 
which uses a similarity transform to isolate the stable/unstable modes, 
and manages switching by first propagating the modes that drive switching and then propagating the remaining modes with switching as an exogenous signal. 
Section \ref{sec:val} presents and discusses the validation of the inversion and stable inversion theory via the 
simulated
\ILILC{} of
a printhead positioning system, the models for which are constructed from historical data of a physical system,
\ref{contribution:sim}.
Finally, Section \ref{sec:conc5} presents conclusions and recommendations for future work.
An exhaustive index of nomenclature can be found in \cite{Spiegel2023}.

\section{System Definition}
\label{sec:sysdef}
A variety of similar 
discrete-time 
\PWA{} system definitions appear in the literature. For the sake of 
familiarity and 
simplicity, the following 
time-varying 
system definition uses a representation similar to \cite{Heemels2001} and \cite{Gao2009}.
\begin{definition}[\PWA{} System] A \PWA{} system is given by
\begin{equation}
\label{eq:PWAdef}
\begin{aligned}
    \state\kdxplus{1} &= A_{q,k}\state\kdx + B_{q,k}\uMIMO\kdx + F_{q,k}
    \\
    \yMIMO\kdx &= C_{q,k}\state\kdx + D_{q,k}\uMIMO\kdx + G_{q,k}
\end{aligned}
\quad \textrm{for }
    \state\kdx 
    \in Q_q
\end{equation}
where $k\in\integer$ is the time-step index, 
$\state\in\real^{n_\state}$, $\uMIMO\in\real^{n_u}$, and $\yMIMO\in\real^{n_y}$ are the state, input, and output vectors,
and $Q_q\in \Qset$ where $\Qset$ is the set of ``locations,'' i.e. a set of disjoint 
regions with union equal to 
$\real^{n_x}$.
Each location $Q_q$ is the union of a set of disjoint 
convex
polytopes.
Here, a
convex
polytope is defined simply as an intersection of half spaces.
Additionally, let the relative degree of the $q$\textsuperscript{th} component model be denoted $\mu_q$ for $q\in\llbracket\, 1,\regionQuant\,\rrbracket$ 
($\llbracket$ $\rrbracket$ indicates a closed 
interval
of integers, 
$|$ $|$ indicates cardinality when applied to a set).
The matrices
$A_{q,k}$, $B_{q,k}$, $F_{q,k}$, $C_{q,k}$, $D_{q,k}$, $G_{q,k}$ 
are 
always 
real.
\label{def:PWAsys}
\end{definition}

To facilitate both mathematical analysis and controller synthesis, the remainder of the 
article 
uses the equivalent closed-form representation
\begin{IEEEeqnarray}{RL}
\eqlabel{eq:sysDefPWA}
\IEEEyesnumber
\IEEEyessubnumber
\label{eq:sysDef_x}
    \state\kdxplus{1} &= 
    \Asum\kdx \state\kdx + \Bsum\kdx u\kdx + \Fsum\kdx
\\
    y\kdx &= \Csum\kdx\state\kdx + \Dsum\kdx u\kdx + \Gsum\kdx
\IEEEyessubnumber 
\label{eq:sysDef_y} 
\end{IEEEeqnarray}
with the upright, bold, capital letter notation defined as
\begin{IEEEeqnarray}{RL}
\Msum\kdx &\coloneq \sumqQ M_{q,k} K_q( \locvec\kdx )
\label{eq:Mshorthand}
\\
K_q(\locvec\kdx) &\coloneq 0^{
\prod_{i=1}^{|\sigvecset_q|}
\norm{\sigvec_{q,i} - \locvec\kdx}
} = 
\begin{cases}
1 & \locvec\kdx\in\sigvecset_q
\\
0 & \textrm{otherwise}
\end{cases}
\label{eq:selector}
\\
\locvec\kdx=\locvec(\state\kdx) &\coloneqq H\left( P\state\kdx - \offsetvec \right)
\label{eq:localization}
\end{IEEEeqnarray}
where $\Msum$ and $M_{q,k}$ stand in for any of
$\{\Asum,\Bsum,\Fsum,\Csum,\Dsum,\Gsum\}$,
and 
$\{ A_{q,k},\allowbreak B_{q,k},\allowbreak F_{q,k},\allowbreak C_{q,k},\allowbreak D_{q,k},\allowbreak G_{q,k} \}$, respectively.
$H$ is the Heaviside step function evaluated element-wise on its vector argument,
$K_q$ is the 
selector function for the $q$\textsuperscript{th} location,
$\sigvecset_q=\{\sigvec_{q,1},\allowbreak\,\sigvec_{q,2},\allowbreak\,\allowbreak\cdots\}$ is the set of binary vector signatures of the $q$\textsuperscript{th} location,
$P\in\real^{n_P\times n_\state}$ is a matrix 
consisting
of concatenated hyperplane orientation vectors, and $\offsetvec\in\real^{n_P}$ is a vector of hyperplane offsets. 
For more information about closed form representations of piecewise defined systems, 
see \cite{Spiegel2019}.

Furthermore, this article assumes:
\begin{enumerate}[label=(A\arabic*),leftmargin=*]
    \item
    \label{C:reachable}
    $\state_0 \in X_0$, where $X_0$ is the set of initial conditions from which all locations $Q_q \in Q$ are reachable in finite 
    time;
    \item 
    \label{C:SISO}
    the system is 
    SISO, $n_u=n_y=1$;
    \item
    \label{C:stateswitch}
    switching depends only on the states, not the input; and
    \item
    \label{C:muc}
    all component models have the same relative degree, $\mu_c$, for all time,
    $\mu_q=\mu_c$ $\forall q\in \llbracket 1,\regionQuant\rrbracket$ and $\forall k$.
    \suspend{listA}
\end{enumerate}

While
assumption \ref{C:stateswitch} is implied by (\ref{eq:localization}), the other assumptions are not implied by the system representation 
\sysDef-(\ref{eq:localization}).

Finally, this work introduces the concept of the 
global dynamical relative degree:
\begin{definition}[Global Dynamical Relative Degree]
\label{def:mu}
The \emph{global dynamical relative degree} of a SISO \PWA{} system
is the smallest number $\dmu\geq 0$ such that 
the explicit expression of $y_{k+\dmu}$ in terms of component state space matrices, selector functions, $\state_k$, and $u_i, \,i\geq k$ contains $u_k$ outside of a selector function for all switching sequences on the interval $\llbracket k,k+\dmu \rrbracket$.
\end{definition}

This $\dmu$ is essentially the traditional relative degree, but neglecting inputs appearing in selector functions. This neglect is introduced to avoid situations in which the only explicit appearance of the input in the output function is within the Heaviside function, leading to a potentially infinite number of input values yielding the same output.
Such non-injectiveness would make inversion unusually challenging. The assumption \ref{C:stateswitch} is made for similar reasons.

\section{Conventional Inversion of \PWA{} Systems}
\label{sec:exactinverse}
This section presents the conventional exact inverses of \PWA{} systems under assumptions \baseAssumptions. 
Conventional inversion is the process of
\begin{enumerate}[label=(Step \arabic*),leftmargin=*]
    \item
    \label{step1}
    obtaining an expression for the previewed output $y\kdxplus{\dmu}$ in terms of $\state\kdx$, $\Msum\kdxplus{\idx}$, and $u\kdxplus{\idx}$ where $\idx\geq 0$,
    \item
    \label{step2}
    solving the previewed output equation for $u\kdx$ in terms of $\state\kdx$, $\Msum\kdxplus{\idx}$, and $u\kdxplus{\idx+1}$, and finally
    \item
    \label{step3}
    taking this expression of $u\kdx$ as the output function of the inverse system, and plugging it in to 
    (\ref{eq:sysDef_x}) to obtain the state transition formula of the inverse system.
\end{enumerate}
Note that $y\kdxplus{\dmu}$ is necessarily an explicit function of $u\kdx$ by Definition \ref{def:mu}.

For systems with $\dmu=0$, this inverse system is unique, and can be expressed explicitly. For systems with $\dmu\geq1$, there may be multiple solutions to the problem of solving $y\kdxplus{\dmu}$ for $u\kdx$ \ref{step2}, and thus the inverse system cannot be expressed explicitly without additional assumptions. This section gives both the general, implicit inverse system for $\dmu\geq 1$ systems and sufficient conditions for the uniqueness of system inversion for $\dmu\in\{1,2\}$ systems along with the corresponding explicit inverse systems.

\subsection{Unique Exact Inversion For 
\texorpdfstring{$\dmu=0$}{Relative Degree 0}
}

\begin{lemma}[Relative Degree of 0]
\label{lem:mu0}
The global dynamical relative degree of a reachable SISO \PWA{} system, i.e. a \PWA{} system satisfying \ref{C:reachable} and \ref{C:SISO}, is 0 if and only if the relative degree of all component models are 0 for all time:
\begin{equation}
\label{eq:mu0}
\mu_q=0 \,\,\,\, \forall q\in \llbracket 1, \regionQuant \rrbracket \text{, } \forall k \iff \dmu=0
\end{equation}
\end{lemma}
\begin{proof}
First, the forward implication is proven directly.
{\allowdisplaybreaks%
\begin{align}
    \mu_q=0 \,\,\,\, 
    \forall q,k
    &\implies 
    D_{q,k} \neq 0 \,\,\,\, 
    \forall q,k
    \\
    D_{q,k} \neq 0 \,\,\,\, 
    \forall q,k
    &\implies
    \Dsum\kdx  \neq 0 \,\,\,\, \forall k
    \\
    \Dsum\kdx  \neq 0 \,\,\,\, \forall k & 
    \implies
    \text{
    $u\kdx$ is explicitly in (\ref{eq:sysDef_y}) $\forall k$
    }
    \\
    \therefore \,\, 
    \mu_q=0 \,\, \forall \,\, 
    &q\in Q 
    \implies 
    \dmu=0
    \label{eq:mu0direct}
\end{align}
}

Now the 
backwards implication is proven by proving the contrapositive.
\begin{align}
    \exists q,k \textrm{ s.t. } \mu_q \neq 0
    &\implies
    \exists q,k \textrm{ s.t. } D_{q,k} = 0
\end{align}
By \ref{C:reachable}, there exists some finite sequence of inputs to bring the system to this location $\Qset_q$ at some time step $k$ such that $\Dsum\kdx  = D_{q,k} = 0$, and  (\ref{eq:sysDef_y}) becomes
\begin{equation}
    y\kdx = \Csum\kdx \state\kdx + \Gsum\kdx
\end{equation}
which is not an explicit function of $u\kdx $, meaning $\dmu \neq 0$.
\begin{equation}
\label{eq:mu0contrapositive}
    \therefore 
    \neg \left(
    \mu_q=0 \,\,\,\, \forall q\in \llbracket 1,\regionQuant \rrbracket, \,\,\forall k
    \right)
    \implies
    \neg \left(
    \dmu=0
    \right)
\end{equation}
where $\neg$ is logical negation.
By (\ref{eq:mu0direct}) and (\ref{eq:mu0contrapositive}), (\ref{eq:mu0}) 
is
true.
\end{proof}

\begin{theorem}[$\dmu=0$ \PWA{} System Inverse]
\label{thm:mu0inv}
The inverse of a \PWA{} system satisfying \baseAssumptions{} with $\mu_c=0$ is itself a \PWA{} system satisfying \baseAssumptions{} and is given by
\begin{IEEEeqnarray}{RL}
\IEEEyesnumber
    \state\kdxplus{1} &= 
    \inv{\Asum}\kdx \state\kdx  + \inv{\Bsum}\kdx y\kdx  + \inv{\Fsum}\kdx 
\IEEEyessubnumber
    \\
    u\kdx  &= \inv{\Csum}\kdx  \state\kdx  + \inv{\Dsum}\kdx  y\kdx  + \inv{\Gsum}\kdx
\IEEEyessubnumber
\end{IEEEeqnarray}
where
\begin{equation*}
\begin{aligned}
    \inv{\Asum}\kdx  &= \Asum\kdx +\Bsum\kdx \inv{\Csum}\kdx 
    & & &
    \inv{\Bsum}\kdx  &= \Bsum\kdx \inv{\Dsum}\kdx
    & & &
    \inv{\Fsum}\kdx  &= \Fsum\kdx +\Bsum\kdx \inv{\Gsum}\kdx 
    \\
    \inv{\Csum}\kdx  &= -\inv{\Dsum}\kdx \Csum\kdx
    & & &
    \inv{\Dsum}\kdx  &= \Dsum^{-1}\kdx
    & & &
    \inv{\Gsum}\kdx  &= -\inv{\Dsum}\kdx  \Gsum\kdx 
\end{aligned}
\end{equation*}
such that $y\kdx$ is the inverse system's input and $u\kdx $ is the output.
\end{theorem}
\begin{proof}
\ref{step1} is satisfied by (\ref{eq:sysDef_y}) and \ref{step2} 
by
\begin{equation}
    u\kdx  = \Dsum\kdx ^{-1}(y\kdx -\Csum\kdx \state\kdx  - \Gsum\kdx )
    \label{eq:umu0}
\end{equation}
Equation (\ref{eq:umu0}) is always well-defined because by Lemma \ref{lem:mu0} ~$\mu_c=0\implies\dmu=0$, which in turn implies $\Dsum\kdx  \neq 0$ $\forall k$, and because $\Dsum$ is always scalar because the system is SISO by \ref{C:SISO}.
Plugging (\ref{eq:umu0}) into (\ref{eq:sysDef_x}) yields the inverse system state transition formula, satisfying \ref{step3}.
\end{proof}
\subsection{Non-unique Exact Inversion For 
\texorpdfstring{$\dmu\geq 1$}{Relative Degree Greater Than 0}
}
For systems with $\dmu\geq 1$, (\ref{eq:sysDef_y}) does not explicitly contain $u\kdx$ because $\Dsum\kdx = 0$ $\forall k$; this is corollary to Lemma \ref{lem:mu0}, Definition \ref{def:mu}, and \ref{C:muc}. Consequently, \ref{step1} necessitates the derivation of an explicit formula for the output preview $y\kdxplus{\dmu}$, i.e. the output at a time after the current time step $k$. This preview of future output is necessary for deriving an output equation that explicitly depends on the current input $u\kdx$.

\begin{lemma}[$\dmu\geq 1$ \PWA{} System Output Preview]
\label{lem:previewOutput}
Given a \PWA{} system satisfying \baseAssumptions{} with known global dynamical relative degree $\dmu$, the output function with minimum preview such that the function is explicitly dependent on an input term outside of the selector functions for any switching sequence is given by
\begin{equation}
\label{eq:previewOutput}
    y\kdxplus{\dmu} = \prevC\kdx \state\kdx + \prevD\kdx u\kdx + \prevG\kdx + \anticaus\kdx(u\kdxplus{1},\cdots,u\kdxplus{\dmu-1})
\end{equation}
with
{\begin{align*}
&\prevC\kdx \coloneq
\Csum\kdxplus{\dmu} \left(\prod_{m=0}^{\dmu-1}\Asum\kdxplus{m}\right)
\qquad
\prevD\kdx \coloneq
\Csum\kdxplus{\dmu}
    \left( \prod_{m=1}^{\dmu-1} \Asum\kdxplus{m} \right) \Bsum\kdx
\\
&\prevG\kdx \coloneq
\Csum\kdxplus{\dmu} \sum_{s=0}^{\dmu-1} \left(
    \left( \prod_{m=s+1}^{\dmu-1} \Asum\kdxplus{m} \right)
    \Fsum\kdxplus{s}
    \right) + \Gsum\kdxplus{\dmu}
\\
&\anticaus\kdx(u\kdxplus{1},\cdots)
\coloneq
\Csum\kdxplus{\dmu}
    \sum_{s=1}^{\dmu-1} 
    \left(
    \left( \prod_{m=s+1}^{\dmu-1} \Asum\kdxplus{m} \right) 
    \Bsum\kdxplus{s} u\kdxplus{s}
    \right)
\end{align*}}
where the factors in the products generated by $\prod$ are ordered sequentially by the index $m$. The factor corresponding to the greatest value of the index must be on the left, and the factor corresponding to the smallest value of the index must be on the right. 
For example,
\begin{equation}
    \prod_{m=0}^{2}\Asum\kdxplus{m} \equiv \Asum\kdxplus{2}\Asum\kdxplus{1}\Asum\kdxplus{0} 
    \not\equiv \Asum\kdxplus{0}\Asum\kdxplus{1}\Asum\kdxplus{2}
\end{equation}
because matrices do not necessarily commute.
Additionally, if the lower bound on the product index exceeds the upper bound on the product index (an ``empty product''), then the product resolves to the identity matrix. Similarly, empty sums resolve to $0$. 
For example
\begin{align}
    \prod_{m=1}^{0}\textup{anything} \equiv I
    \qquad \qquad
    \sum_{s=1}^{0}\textup{anything} \equiv 0
\end{align}
\end{lemma}
\begin{proof}
Let the base case of the proof by induction be $\dmu=1$ such that
\begin{equation}
\label{eq:BaseCase}
    y\kdxplus{1} = \Csum\kdxplus{1} \left(
    \Asum\kdx  \state\kdx 
    +
    \Fsum_{k}
    +
    \Bsum\kdx  u\kdx 
    \right)
    +
    \Gsum\kdxplus{1}
\end{equation}
which is achieved equivalently from (\ref{eq:previewOutput}) and from the system definition by plugging (\ref{eq:sysDef_x}) into (\ref{eq:sysDef_y}) incremented by one time step (i.e. plugging the equation for $\state\kdxplus{1}$ into the equation for $y\kdxplus{1}$).
The preview is minimal because, by Definition \ref{def:mu}, $\Csum\kdxplus{1}\Bsum\kdx \neq 0$ for all switching sequences on $\llbracket k,k+1 \rrbracket$, and by Lemma \ref{lem:mu0} $\dmu\geq 1 \implies \Dsum\kdx =0$ $\forall k$. 
In other words, for $\dmu=1$, (\ref{eq:BaseCase}) is always an explicit function of $u\kdx$. Thus, regardless of switching sequence, outputs further in the future, such as $y\kdxplus{2}$, never need to be considered in order to explicitly relate the current input to an output.

Then consider (\ref{eq:previewOutput}) with $\dmu=\nu$ as the foundation of the induction step. To prove (\ref{eq:previewOutput}) holds for $\dmu=\nu+1$, first increment (\ref{eq:previewOutput}) with $\dmu=\nu$ by one time step, yielding
\begin{align}
    \nonumber
    y\kdxplus{\nu+1} 
    &
    =
    \Csum\kdxplus{\nu+1}\left( \prod_{m=0}^{\nu-1}\Asum\kdxplus{1+m}   \right)\state\kdxplus{1}
    \\
    \nonumber
    & \quad
    +
    \Csum\kdxplus{\nu+1}\left(\prod_{m=1}^{\nu-1}\Asum\kdxplus{1+m}\right)\Bsum\kdxplus{1}u\kdxplus{1}
    \\
    \nonumber
    & \quad
    +
    \Csum\kdxplus{\nu+1}\sum_{s=0}^{\nu-1}\left( \left(\prod_{m=s+1}^{\nu-1}\Asum\kdxplus{1+m}\right)\Fsum\kdxplus{1+s} \right) + \Gsum\kdxplus{1+\nu}
    \\
    & \quad
    +
    \Csum\kdxplus{\nu+1}\sum_{s=1}^{\nu-1}\left(\left(\prod_{m=s+1}^{\nu-1}\Asum\kdxplus{1+m}\right) \Bsum\kdxplus{1+s}u\kdxplus{1+s}\right)
\end{align}
This 
is
simplified by first factoring out $\Csum\kdxplus{\nu+1}$ and adjusting the product $\Pi$ indices to subsume the constant $+1$ in $\Asum\kdxplus{1+m}$:
\begin{align}
\hspace{-0.02in}
    \nonumber
    y\kdxplus{\nu+1} 
    &
    =
    \Csum\kdxplus{\nu+1}
    \left(\left( \prod_{m=1}^{\nu}\Asum\kdxplus{m}   \right)\state\kdxplus{1}
    +
    \left(\prod_{m=2}^{\nu}\Asum\kdxplus{m}\right)\Bsum\kdxplus{1}u\kdxplus{1}
    \right.
    \\
    \nonumber
    & \quad
    \left.
    +
    \sum_{s=0}^{\nu-1}\left( \left(\prod_{m=s+2}^{\nu}\Asum\kdxplus{m}\right)\Fsum\kdxplus{1+s} \right) 
    \right.
    \\
    & \quad
    \left.
    +
    \sum_{s=1}^{\nu-1}\left(\left(\prod_{m=s+2}^{\nu}\Asum\kdxplus{m}\right) \Bsum\kdxplus{1+s}u\kdxplus{1+s}\right)
    \right)
    +
    \Gsum\kdxplus{\nu+1}
\end{align}
Similarly, the sum $\sum$ indices may be adjusted to subsume the constant $+1$ in $\Fsum\kdxplus{1+s}$, $\Bsum\kdxplus{1+s}$, and $u\kdxplus{1+s}$. Because the sum index $s$ also appears in the lower bound of the products, $m=s+2$, the product index lower bound must also be adjusted.
\begin{align}
\nonumber
    y\kdxplus{\nu+1} 
    &
    =
    \Csum\kdxplus{\nu+1}
    \left(\left( \prod_{m=1}^{\nu}\Asum\kdxplus{m}   \right)\state\kdxplus{1}
    +
    \left(\prod_{m=2}^{\nu}\Asum\kdxplus{m}\right)\Bsum\kdxplus{1}u\kdxplus{1}
    \right.
    \\
    \nonumber
    & \qquad
    \left.
    +
    \sum_{s=1}^{\nu}\left( \left(\prod_{m=s+1}^{\nu}\Asum\kdxplus{m}\right)\Fsum\kdxplus{s} \right) 
    \right.
    \\
    & \qquad
    \left.
    +
    \sum_{s=2}^{\nu}\left(\left(\prod_{m=s+1}^{\nu}\Asum\kdxplus{m}\right) \Bsum\kdxplus{s}u\kdxplus{s}\right)
    \right)
    +
    \Gsum\kdxplus{\nu+1}
\end{align}
Finally, the two input terms (those containing $u$, arising from $\prevD\kdxplus{1}$ and $\anticaus\kdxplus{1}$) can be combined to achieve
\begin{align}
    \nonumber
    y\kdxplus{\nu+1} 
    &
    = 
    \Csum\kdxplus{\nu+1}
    \left(
    \left(\prod_{m=1}^{\nu}\Asum\kdxplus{m}\right) 
    \state\kdxplus{1}
    \right.
    \\
    \nonumber & \qquad
    \left.
    +
    \sum_{s=1}^{\nu} \left(
    \left( \prod_{m=s+1}^{\nu} \Asum\kdxplus{m} \right)
    \Fsum\kdxplus{s}
    \right)
    \right.
    \\
    & \qquad
    \left.
    +
    \sum_{s=1}^{\nu} 
    \left(
    \left( \prod_{m=s+1}^{\nu} \Asum\kdxplus{m} \right) 
    \Bsum\kdxplus{s} u\kdxplus{s}
    \right)
    \right)
    +
    \Gsum\kdxplus{\nu+1}
    \label{eq:lem2increment}
\end{align}
which is a function of $\state\kdxplus{1}$ and potentially $u\kappadx$ for $\kappa \in~ \llbracket k+~1,k+\nu+1 \rrbracket$.
The dependence of $y\kdxplus{\nu+1}$ on the input terms is conditioned on the switching sequence. Definition \ref{def:mu} implies that if $\dmu=\nu+1$ there exists some switching sequence on $\llbracket k,k+\nu+1 \rrbracket$ such that the input coefficients in (\ref{eq:lem2increment}) are zero, i.e.
\begin{multline}
\exists\, \{\state\kappadx\, |\, \kappa \in \llbracket k,k+\nu+1\rrbracket\} 
\quad
\text{ s.t. }
\quad
\\
\forall s\in\llbracket 1,\nu\rrbracket \quad
    \Csum\kdxplus{\nu+1}
    \left( \prod_{m=s+1}^{\nu} \Asum\kdxplus{m} \right) 
    \Bsum\kdxplus{s}
    =
    0
\end{multline}
Thus, to guarantee the expression for $y\kdxplus{\nu+1}$ explicitly contains the input for all switching sequences,
$\state\kdxplus{1}$ in (\ref{eq:lem2increment}) must be expanded (via (\ref{eq:sysDef_x})) to be in terms of $\state\kdx $ and $u\kdx $ explicitly.
The resulting expression can be rearranged as follows:
\begin{multline}
    y\kdxplus{\nu+1} = 
    \Csum\kdxplus{\nu+1}
    \left(
    \left(\prod_{m=1}^{\nu}\Asum\kdxplus{m}\right) 
    \left(
    \Asum\kdx  \state\kdx  + \Bsum\kdx  u\kdx  + \Fsum\kdx 
    \right)
    \right.
    \\
    \left.
    +
    \sum_{s=1}^{\nu} \left(
    \left( \prod_{m=s+1}^{\nu} \Asum\kdxplus{m} \right)
    \Fsum\kdxplus{s}
    \right)
    \right)
    +
    \anticaus\kdx(u\kdxplus{1},\cdots,u\kdxplus{\nu})
    +
    \Gsum\kdxplus{\nu+1}
\end{multline}
\begin{align}
    \nonumber
    y\kdxplus{\nu+1} 
    &
    = 
    \Csum\kdxplus{\nu+1}
    \left(
    \left(\prod_{m=1}^{\nu}\Asum\kdxplus{m}\right) 
    \Asum\kdx  \state\kdx 
    +
    \left(\prod_{m=1}^{\nu}\Asum\kdxplus{m}\right)
    \Bsum\kdx  u\kdx 
    \right.
    \\
    \nonumber & \qquad
    \left.
    +
    \sum_{s=1}^{\nu} \left(
    \left( \prod_{m=s+1}^{\nu} \Asum\kdxplus{m} \right)
    \Fsum\kdxplus{s}
    \right)
    +
    \left(\prod_{m=1}^{\nu}\Asum\kdxplus{m}\right) \Fsum\kdx 
    \right.
    \\
    & \qquad
    \left.
    +
    \Gsum\kdxplus{\nu+1}
    \right)
    +\anticaus\kdx(u\kdxplus{1},\cdots,u\kdxplus{\nu})
\end{align}
\begin{multline}
    y\kdxplus{\nu+1} = 
    \Csum\kdxplus{\nu+1}
    \left(
    \left(\prod_{m=0}^{\nu}\Asum\kdxplus{m}\right) 
    \state\kdx 
    +
    \left(\prod_{m=1}^{\nu}\Asum\kdxplus{m}\right)
    \Bsum\kdx  u\kdx
    \right.
    \\
    \left.
    +
    \sum_{s=0}^{\nu} \left(
    \left( \prod_{m=s+1}^{\nu} \Asum\kdxplus{m} \right)
    \Fsum\kdxplus{s}
    \right)
    \right)
    +
    \Gsum\kdxplus{\nu+1}
    +\anticaus\kdx(u\kdxplus{1},\cdots,u\kdxplus{\nu})
\label{eq:lem2conc}
\end{multline}
Equation (\ref{eq:lem2conc}) is 
equation (\ref{eq:previewOutput}) for $\dmu=\nu+1$, thereby proving the lemma.
\end{proof}

Using the output preview equation (\ref{eq:previewOutput}), a general \PWA{} system inverse for $\dmu\geq 1$ can be found in the same manner as for $\dmu=0$.

\begin{theorem}[General $\dmu\geq 1$ \PWA{} System Inverse]
\label{thm:genInv}
Given any \PWA{} system satisfying \ref{C:reachable}-\ref{C:muc} with known global dynamical relative degree $\dmu \geq ~1$, the inverse system with $u\kdx $ as the output is given by the implicit, anticausal system
\begin{IEEEeqnarray}{RL}
\nonumber
    \state\kdxplus{1} &= 
    \inv{\Asum}\kdx  \state\kdx  
    +
    \inv{\Bsum}\kdx y\kdxplus{\dmu}
    + \inv{\Fsum}\kdx
    \\
\IEEEyesnumber
    & \qquad
    -\Bsum\kdx\prevD\kdx^{-1}\anticaus\kdx\left(u\kdxplus{1},\cdots,u\kdxplus{\dmu-1}\right)
\IEEEyessubnumber\label{eq:xGen}
    \\
    \nonumber
    u\kdx  &=  \inv{\Csum}\kdx\state\kdx+\inv{\Dsum}\kdx y\kdxplus{\dmu}+\inv{\Gsum}\kdx
    \\
    & \qquad
    -\prevD\kdx^{-1}\anticaus\kdx\left(u\kdxplus{1},\cdots,u\kdxplus{\dmu-1}\right)
\IEEEyessubnumber\label{eq:uGen}
\end{IEEEeqnarray}
where
\begin{align*}
\inv{\Asum}\kdx &= \Asum\kdx + \Bsum\kdx\inv{\Csum}\kdx
&
\inv{\Bsum}\kdx &= \Bsum\kdx\inv{\Dsum}\kdx
&
\inv{\Fsum}\kdx &= \Fsum\kdx+\Bsum\kdx\inv{\Gsum}\kdx
\\
\inv{\Csum}\kdx &= -\inv{\Dsum}\kdx\prevC\kdx
&
\inv{\Dsum}\kdx &= \prevD\kdx^{-1}
&
\inv{\Gsum}\kdx&=-\inv{\Dsum}\kdx\prevG\kdx
\end{align*}
\end{theorem}
\begin{proof}
The sole term in
Lemma \ref{lem:previewOutput}'s
(\ref{eq:previewOutput})
containing $u\kdx $ outside of a selector function 
is $\prevD\kdx u\kdx$. The coefficient $\prevD\kdx=~\Csum\kdxplus{\dmu}\left( \prod_{m=1}^{\dmu-1} \Asum\kdxplus{m} \right) \Bsum_{k}$ is always scalar because the system is SISO, \ref{C:SISO}, and always nonzero by Definition \ref{def:mu}.
Thus (\ref{eq:previewOutput}) can be divided by $\prevD\kdx$ and $u\kdx $ can be arithmetically 
isolated on one side of the equation,
yielding (\ref{eq:uGen}).

Equation (\ref{eq:uGen}) is implicit in general because (\ref{eq:previewOutput}) cannot be uniquely solved for $u\kdx$ in some cases.
This is proven by presenting an example in which multiple input trajectories have the same output trajectory. Consider the two-location system
\begin{equation}
\begin{aligned}
    &
    A_{1,k}=\begin{bmatrix}0&1\\0&0\end{bmatrix}
    \quad
    A_{2,k}=\begin{bmatrix}0&2\\0&0\end{bmatrix}
    \quad
    B_{1,k}=B_{2,k}=\begin{bmatrix}0\\1\end{bmatrix}
    \\
    &
    C_{1,k}=C_{2,k}=\begin{bmatrix}1&0\end{bmatrix}
    \\
    &
    P=\begin{bmatrix}0&1\end{bmatrix}
    \qquad
    \offsetvec=1.5
    \qquad
    \sigvecset_1=\{1\}
    \qquad
    \sigvecset_2=\{2\}
\end{aligned}
\end{equation}
with $F$, $D$, and $G$ matrices all equal to zero. Given $\state\kdx=~[\,0,\,\,0\,]^T$, both $u\kdx=2$ and $u\kdx=1$ yield $y\kdxplus{\dmu}=2$. Because there is not a unique solution to (\ref{eq:previewOutput}) for $u\kdx$, there does not exist an explicit formula for the solution (\ref{eq:uGen}).

The system is necessarily anticausal because $u\kdx $ is necessarily a function of $y\kdxplus{\dmu}$ 
(and no prior $y$ terms)
and $\dmu > 0$.
\end{proof}

\begin{remark}[Inverse Implicitness]
Analytically, the implicitness of (\ref{eq:uGen}) arises in $\inv{\Csum}\kdx$ through $\prevC\kdx$, which is a function of $\Csum\kdxplus{\dmu}$ by definition, 
and $\Csum\kdxplus{\dmu}$
is a function of $\state\kdxplus{\dmu}$ by 
(\ref{eq:Mshorthand})-(\ref{eq:localization}). 
Finally, $\state\kdxplus{\dmu}$ is a function of $u\kdx$ via
\begin{multline}
    \state\kdxplus{\dmu} = 
    \left(\prod_{m=0}^{\dmu-1}\Asum\kdxplus{m}\right) 
    \state\kdx 
    +
    \sum_{s=0}^{\dmu-1} \left(
    \left( \prod_{m=s+1}^{\dmu-1} \Asum\kdxplus{m} \right)
    \Fsum\kdxplus{s}
    \right)
    +
    \\
    \sum_{s=1}^{\dmu-1} 
    \left(
    \left( \prod_{m=s+1}^{\dmu-1} \Asum\kdxplus{m} \right) 
    \Bsum\kdxplus{s} u\kdxplus{s}
    \right)
    +
    \left( \prod_{m=1}^{\dmu-1} \Asum\kdxplus{m} \right) 
    \Bsum_{k} u_{k}
\end{multline}
following from $y\kdxplus{\dmu}=\Csum\kdxplus{\dmu}\state\kdxplus{\dmu}+\Gsum\kdxplus{\dmu}$ and Lemma \ref{lem:previewOutput}.
\end{remark}

\begin{remark}[Input Preview and Inter-location Relative Degree]
Note $\anticaus\kdx$ is written as a function of a set of previewed
$u$-values.
The written set of $u$-values is the maximum quantity of $u$-values that may be required by $\anticaus\kdx$. Depending on the switching sequence, fewer previewed $u$-values may be required. 
In fact, by the definition of global dynamical relative degree $\dmu$, there must exist a switching sequence for which no previewed $u$-values are required, because otherwise $\dmu$ would be smaller. In other words, Definition \ref{def:mu} and Lemma \ref{lem:previewOutput} imply
\begin{multline}
\label{eq:uPreviewless}
    \exists\{\state\kappadx \, |\,  \kappa\in\llbracket k,k+\dmu\rrbracket\} \text{ s.t. }
    \\
    \forall s\in\llbracket 1,\dmu-1\rrbracket \,\,\,
    \Csum\kdxplus{\dmu}
    \left(\prod_{m=s+1}^{\dmu
    -
    1}\Asum\kdxplus{m}\right)\Bsum\kdxplus{s}=0 
    \\
    \land \,\,\, \Csum\kdxplus{\dmu}\left(\prod_{m=1}^{\dmu-1}\Asum\kdxplus{m}\right)\Bsum\kdx\neq 0
\end{multline}

For affine time-invariant systems without piecewise definition, there is never required input preview because the expressions claimed equal to zero in (\ref{eq:uPreviewless}) reduce as
\begin{equation}
    \Csum\kdxplus{\dmu}
    \left(\prod_{m=s+1}^{\dmu
    -
    1}\Asum\kdxplus{m}\right)\Bsum\kdxplus{s}
    =
    CA^{\dmu-1-s}B
    \qquad
    s\in\llbracket 1,\,\dmu-1\rrbracket
\end{equation}
which are always zero regardless of state sequence.
However, it is important to emphasize that this is not the case for \PWA{} systems.
Even when the relative degrees of all component models 
are equal to $\mu_c$, \ref{C:muc}, the inter-location relative degree may not be equal to $\mu_c$.
More formally
\begin{align}
    &C_{1,k+1}B_{1,k} = 0 \land C_{2,k+1}B_{2,k} = 0 \centernot\implies C_{1,k+1}B_{2,k}=0
    \end{align}
    and
    \begin{align}
    &C_{1,k+1}B_{1,k} \neq 0 \land C_{2,k+1}B_{2,k} \neq 0 \centernot\implies C_{1,k+1}B_{2,k}\neq0
\end{align}
and this lack of conclusiveness regarding the inter-location relative degree generalizes to larger $\mu_c$. In short, the inter-location relative degree may be either lower or higher than $\mu_c$, and may be different for different switching sequences (with a maximum value of $\dmu$ as given by Definition \ref{def:mu}).
\end{remark}

Because of the implicitness and potential requirement for input preview in the general \PWA{} system inverse,
it is nontrivial to use it for computing input trajectories from output trajectories.
However, there are conditions under which inversion of a \PWA{} system with $\dmu\geq 1$ is unique, and the inverse itself becomes an explicit \PWA{} system as is the case for $\dmu=0$.
The remainder of this section provides such sufficient conditions for the cases of $\dmu=1$ and $\dmu=2$, with the corresponding explicit inverse formulas given as corollaries to Theorem \ref{thm:genInv}.
\subsection{Unique Exact Inverses For 
\texorpdfstring{$\dmu\in\{1,2\}$}{Relative Degree 1 or 2}
}
First the 
``location-invariant output function''
assumption is introduced:
\begin{enumerate}[label=(A\arabic*),leftmargin=*]
    \resume{listA}
    \item 
    \label{C:locIndep}
    $\Csum\kdx=C\kdx$, $\Dsum\kdx=D\kdx$, $\Gsum\kdx=G\kdx$, with $C\kdx$, $D\kdx$, $G\kdx$ indicating parameters that potentially vary with time but that are identical $\forall q\in\llbracket 1,\regionQuant\rrbracket$
    and are known $\forall k$
    \suspend{listB}
\end{enumerate}

\begin{lemma}[Relative Degree of 1]
\label{lem:mu1}
A \PWA{} system satisfying \ref{C:reachable}-\ref{C:locIndep} has a global dynamical relative degree of 1 if and only if the relative degree of all component models are 1 for all time:
\begin{equation}
    \mu_c=1 \iff \dmu=1
    \label{eq:LemMu1}
\end{equation}
\end{lemma}
\begin{proof}
The forward implication follows directly from \ref{C:muc}, \ref{C:locIndep}, and $\mu_c=1$:
    \begin{align}
        \mu_c=1 &\implies D\kdxplus{1}=0 
        \,\land\, C\kdxplus{1}B_{q,k}\neq 0 \,\,\,\forall q\in\llbracket 1,\regionQuant\rrbracket
        \\
        &\implies
        \Dsum\kdxplus{1}=0 \,\land\, \Csum\kdxplus{1}\Bsum\kdx\neq 0 \,\,\,\forall k
    \end{align}
This is equivalent to implying that the inter-location 
relative degree
can be neither higher nor lower than $\mu_c=1$ under \ref{C:muc} and \ref{C:locIndep}
\begin{equation}
    \therefore \,\,\, \mu_c=1\implies \dmu=1
\end{equation}
The backward implication follows from (\ref{eq:uPreviewless}) (i.e. the combination of Definition \ref{def:mu} and Lemma \ref{lem:previewOutput}), which implies
\begin{align}
    \dmu=1\implies \exists\{\state\kdx,\state\kdxplus{1}\} \text{ s.t. } \Csum\kdxplus{1}\Bsum\kdx\neq 0
\end{align}
Then, by \ref{C:muc}, \ref{C:locIndep}, and the fact that $\dmu=1\implies \Dsum\kdx=0$ by Lemma \ref{lem:mu0},
one finds that
$\Csum\kdxplus{1}\Bsum\kdx=C\kdxplus{1}\Bsum\kdx\neq 0\implies \mu_c=1$. Therefore (\ref{eq:LemMu1}) is true.
\end{proof}

\begin{corollary}[Unique Inverse of $\dmu=1$ \PWA{} Systems]
\label{corollary:mu1}
The inverse of a \PWA{} system satisfying \ref{C:reachable}-\ref{C:locIndep} with 
$\mu_c=1$
is 
given by the following explicit, anticausal \PWA{} system.
\begin{IEEEeqnarray}{RL}
\IEEEyesnumber
    \state\kdxplus{1} &= 
    \inv{\Asum}\kdx \state\kdx  + \inv{\Bsum}\kdx y\kdxplus{1} + \inv{\Fsum}\kdx 
\IEEEyessubnumber
    \\
    u\kdx  &= \inv{\Csum}\kdx  \state\kdx  + \inv{\Dsum}\kdx  y\kdxplus{1} + \inv{\Gsum}\kdx
\IEEEyessubnumber
\end{IEEEeqnarray}
where
\begin{align*}
    \inv{\Asum}\kdx  &= \Asum\kdx +\Bsum\kdx \inv{\Csum}\kdx
    & & &
    \inv{\Csum}\kdx  &= - \inv{\Dsum}\kdx C\kdxplus{1}\Asum\kdx 
    \\
    \inv{\Bsum}\kdx  &= \Bsum\kdx \inv{\Dsum}\kdx
    & & &
    \inv{\Dsum}\kdx  &= \left( C\kdxplus{1} \Bsum_{k}\right)^{-1}
    \\
    \inv{\Fsum}\kdx  &= \Fsum\kdx +\Bsum\kdx\inv{\Gsum}\kdx
    & & &
    \inv{\Gsum}\kdx  &= -\inv{\Dsum}\kdx\left(C\kdxplus{1} \Fsum\kdx  + G\kdxplus{1} \right)
\end{align*}
\end{corollary}
\begin{proof}
$\dmu=1$ by Lemma \ref{lem:mu1}.
Plugging \ref{C:locIndep} and $\dmu=1$ into (\ref{eq:uGen}) 
yields
\begin{equation}
    u\kdx  = (C\kdxplus{1}\Bsum\kdx )^{-1} (y\kdxplus{1}-C\kdxplus{1}(\Asum\kdx \state\kdx +\Fsum\kdx )-G\kdxplus{1})
\end{equation}
which is explicit.
Note that 
$\left(C\kdxplus{1}\Bsum\kdx\right)^{-1}$ is always well 
defined because it is scalar by \ref{C:SISO} and is nonzero
by $\dmu=1$ and Definition \ref{def:mu}.
\end{proof}

Derivation of explicit inverses for systems with $\dmu=2$ uses ``output function time-invariance'' and ``output-based switching'' assumptions. If $\dmu>0$, this can be expressed as
\begin{enumerate}[label=(A\arabic*),leftmargin=*]
\resume{listB}
\item
\label{C:outputSwitch}
$P = P_o\Csum\kdx$,
\,
$\offsetvec = 
\offsetvec_o-P_o\Gsum\kdx$
\,with
$\Csum\kdx=C$,
$\Gsum\kdx=G$ $\forall k$
\suspend{listC}
\end{enumerate}
where $P_o$ and $\offsetvec_o$ contain the orientation vectors and offsets of hyperplanes in the output space $\mathbb{R}^{n_y}$.

\begin{corollary}[Unique Inverse of $\dmu=2$ \PWA{} Systems]
The inverse of a \PWA{} system satisfying \ref{C:reachable}-\ref{C:outputSwitch} with known global dynamical relative degree $\dmu=2$ is given by the following explicit, anticausal \PWA{} system.
\begin{IEEEeqnarray}{RL}
\IEEEyesnumber
    \state\kdxplus{1} &= 
    \inv{\Asum}\kdx \state\kdx  + \inv{\Bsum}\kdx y\kdxplus{2} + \inv{\Fsum}\kdx 
\IEEEyessubnumber
    \\
    u\kdx  &= \inv{\Csum}\kdx  \state\kdx  + \inv{\Dsum}\kdx  y\kdxplus{2} + \inv{\Gsum}\kdx 
\IEEEyessubnumber
\end{IEEEeqnarray}
where
\begin{IEEEeqnarray}{C}
    \inv{\Asum}\kdx  = 
    \Asum\kdx +\Bsum\kdx 
    \inv{\Csum}\kdx 
    \qquad\quad
    \inv{\Bsum}\kdx  = \Bsum\kdx 
    \inv{\Dsum}\kdx 
    \qquad\quad
    \inv{\Fsum}\kdx  = \Fsum\kdx +\Bsum\kdx 
    \inv{\Gsum}\kdx 
\nonumber
\\[1pt]
    \inv{\Csum}\kdx  = 
    -
    \inv{\Dsum}\kdx
    C\kdxplus{2}\Asum\kdxplus{1}\Asum\kdx 
    \qquad\quad
    \inv{\Dsum}\kdx  = 
    \left( C\kdxplus{2} \Asum\kdxplus{1}\Bsum\kdx \right)^{-1}
\nonumber
\\[1pt]
    \inv{\Gsum}\kdx  = 
    -
    \inv{\Dsum}\kdx
    \left(
    C\kdxplus{2} \Asum\kdxplus{1}\Fsum\kdx  + C\kdxplus{2}\Fsum\kdxplus{1} + G\kdxplus{2} \right)
\nonumber
\end{IEEEeqnarray}
\end{corollary}
\begin{proof}
Plugging \ref{C:locIndep} and $\dmu=2$ into (\ref{eq:uGen}) yields
\begin{multline}
        u\kdx =
        \left(
        C\kdxplus{2} \Asum\kdxplus{1} \Bsum\kdx 
        \right)^{-1}
        \left(
        y\kdxplus{2} - 
        C\kdxplus{2}
        \Asum\kdxplus{1}\Asum\kdx 
        \state\kdx 
        \right.
        \\
        \left.
        -
        C\kdxplus{2}
        \Asum\kdxplus{1}
        \Fsum\kdx 
        -
        C\kdxplus{2}
        \Fsum\kdxplus{1}
        -
        G\kdxplus{2}
        \right)
        \label{eq:mu2ass}
\end{multline}
In general, (\ref{eq:mu2ass}) would be implicit because of $\Asum\kdxplus{1}$ and $\Fsum\kdxplus{1}$'s dependence on $\state\kdxplus{1}$, and thus $u\kdx $, via the selector functions
\begin{equation}
\begin{aligned}
    K_q(\locvec(\state\kdxplus{1})) &= 
    0^{
    \prod_{i=1}^{\sigvecset_q}
    \norm{\sigvec_{q,i}-H\left(
    P\state\kdxplus{1}
    - \offsetvec
    \right)}}
    \\
    &= 
    0^{
    \prod_{i=1}^{\sigvecset_q}
    \norm{\sigvec_{q,i}-H\left(
    P\left(
    \Asum\kdx \state\kdx  + \Bsum\kdx u\kdx  + \Fsum\kdx 
    \right)
    - \offsetvec
    \right)}}
\end{aligned}
\label{eq:implicitK}
\end{equation}
However, 
under
\ref{C:outputSwitch} (in combination with \ref{C:locIndep}) this becomes
\begin{equation}
\hspace{-0.03in}
    K_q(\locvec(\state\kdxplus{1})) = 
    0^{
    \prod_{i=1}^{\sigvecset_q}
    \norm{\sigvec_{q,i}-H\left(
    P_oC\kdxplus{1}\left(
    \Asum\kdx \state\kdx  + \Fsum\kdx 
    \right)
    - \offsetvec_o + P_oG\kdxplus{1}
    \right)}}
    \label{eq:explicitK}
\end{equation}
which is not a function of $u\kdx $ and is thus explicit.

The reduction of (\ref{eq:implicitK}) to (\ref{eq:explicitK}) relies on the fact that $C\kdxplus{1}\Bsum\kdx ~=0$ $\forall k$. This is true for $\dmu=2$ systems under \ref{C:locIndep} because $\dmu=2\implies\mu_c>1$ by Lemma \ref{lem:mu1}.
\end{proof}

\begin{remark}[$\mu_c$, $\dmu$ Relationship]
Note that unlike for relative degrees of 0 and 1, $\mu_c=2\centernot\implies\dmu=2$ under assumptions \ref{C:reachable}-\ref{C:outputSwitch}. If $n_x>2$, there exists systems for which $\mu_c=2$ but $\dmu>2$ due to inter-location dynamics.
\end{remark}

\begin{remark}[Time Invariance]
Relaxation of 
``output function time-invariance'' 
is often
practically admissible (i.e. replacing $C$, $G$, with potentially time-varying $C\kdx$, $G\kdx$ in \ref{C:outputSwitch}).
Time-invariance is assumed purely for formal satisfaction of Definition \ref{def:PWAsys} and (\ref{eq:localization}), 
which use constant state space partitioning.
The relaxation uses constant output space partitioning but allows time-varying $P$, $\offsetvec$. 
This requires changes to definitions, but not 
to proofs 
(consider that Definition \ref{def:PWAsys} already allows much time-varying switching because time can be made into a state).
Conversely, \ref{C:locIndep} (\emph{location}-invariance) ensures 
the \ref{C:outputSwitch} equations are explicit,
and may not be relaxable.
\end{remark}

\newpage
\section{Stable Inversion of \PWA{} Systems}
\label{sec:stabinv}
\subsection{Exact Stable Inversion}

For many \PWA{} systems, propagating the 
inverse system
forward in time from an initial state 
at time $k=k_0$ 
with a bounded reference $y\kdxplus{\dmu}=r\kdxplus{\dmu}$
will yield a 
$u\kdx$ trajectory that is bounded for all $k\geq k_0$
and suitable for feedforward control.
However, inverse 
system instabilities may arise from NMP component dynamics, 
causing $u\kdx$ to become unbounded under this conventional 
propagation,
despite the bounded reference.
In such cases, a bounded $u\kdx$ may still be achievable on a bi-infinite timeline via stable inversion.
Formally, the stable inversion problem may be given as follows.
\begin{definition}[\PWA{} Stable Inversion Problem Statement]
\label{def:SIproblem}
Given an explicit inverse \PWA{} system representation
\begin{IEEEeqnarray}{RL}
\eqlabel{eq:inverse}
\IEEEyesnumber
\state\kdxplus{1} &= \inv{\Asum}\kdx\state\kdx + \inv{\Bsum}\kdx y\kdxplus{\dmu} + \inv{\Fsum}\kdx
\IEEEyessubnumber\label{eq:inverseState}
\\
u\kdx &=\inv{\Csum}\kdx\state\kdx + \inv{\Dsum}\kdx y\kdxplus{\dmu}  +\inv{\Gsum}\kdx
\IEEEyessubnumber\label{eq:inverseOutput}
\end{IEEEeqnarray}
and a 
bounded
reference trajectory $y\kdxplus{\dmu}=r\kdxplus{\dmu}$ known for all $k\in\integer$, a two point boundary value problem is formed by (\ref{eq:inverseState}) and the boundary conditions $\state_{-\infty}=\state_\infty=0$.
The solution to the stable inversion problem is the bounded bi-infinite time series $u\kdx\in\real$ $\forall k$,
which is
generated by (\ref{eq:inverseOutput}) and the bounded bi-infinite solution $\state\kdx$ to the boundary value problem.
\end{definition}

The following assumptions on system parameter boundedness and boundary conditions are common in some form across much stable inversion literature.
\begin{enumerate}[label=(A\arabic*),leftmargin=*]
\resume{listC}
\item
\label{C:boundedSystem}
There exists a supremum to the norms of the 
inverse system's output function matrices:
\begin{equation}
\sup_k\norm{\inv{\Csum}\kdx}
 \in \real
 \quad
 \sup_k\norm{\inv{\Dsum}\kdx}
 \in \real
 \quad
 \sup_k\norm{\inv{\Gsum}\kdx}
 \in \real
\end{equation}
Any vector norm may be 
used. The
matrix norm is that induced by the vector norm 
(i.e. the operator norm).

\item
\label{C:refdecay}
The forcing 
part of the inverse's state function decays
to zero at the extremities of the bi-infinite time series:
\begin{multline}
    \hspace{-0.125in}
    \forall \varepsilon \in \real_{>0}
    \quad
    \exists \eta_1, \eta_2 \in\integer
    \quad
    \text{ s.t. }
    \\
    \hspace{-1in}
    \norm{\inv{\Bsum}\kdx y\kdxplus{\dmu} + \inv{\Fsum}\kdx}
    <\varepsilon
    \quad
    \forall k\in(-\infty,\eta_1\rrbracket\cup\llbracket \eta_2,\infty)
\end{multline}
\suspend{listD}
\end{enumerate}

Additionally, stable inversion of \PWA{} systems involves two challenges not faced 
with
linear systems. First, 
the dynamics of all locations and 
the inter-location dynamics must
be
simultaneously 
accounted for when decoupling the stable and unstable system modes. 
Second, there must be a way to manage switching in the two partial system evolutions.
These challenges are manifested in 
the following assumptions.
\begin{enumerate}[label=(A\arabic*),leftmargin=*]
\resume{listD}
\item
\label{C:decouplable}
There exists a similarity transform matrix $V\in\real^{n_x\times n_x}$ that decouples the stable and unstable modes of 
(\ref{eq:inverseState}).
Formally this decoupling can be expressed as
\begin{align}
    V \inv{\Asum}\kdx V^{-1}=\begin{bmatrix}[1.5] 
    \Asumdc^\stab\kdx & 0_{n_\stab \times n_\unstab}
    \\
    0_{n_\unstab\times n_\stab} & \Asumdc^\unstab\kdx
    \end{bmatrix}
    \quad \forall k
    \label{eq:decouplable}
\end{align}
where 
$0_{\circ \times \bullet}$ is a $\circ$-by-$\bullet$ matrix of zeros,
$n_\stab$ is the number of stable modes, $n_\unstab$ is the number of unstable modes, $n_\stab+n_\unstab=n_x$,
$\Asumdc^\unstab\kdx$ has all eigenvalue magnitudes $>1$ $\forall k$,
and the free systems
\begin{align}
    \chi\kdxplus{1}^\stab = \decoup{\Asum}\kdx^\stab \chi\kdx^\stab
    \qquad
    \chi\kdxplus{1}^\unstab = \left(\decoup{\Asum}\kdx^\unstab\right)^{-1}\chi\kdx^\unstab
    \label{eq:free}
\end{align}
with 
appropriately 
sized
state vectors $\chi^\stab$, $\chi^\unstab$ 
are globally uniformly asymptotically stable about the 
origin. (See \cite{Chen2020} for theorems on sufficient conditions for uniform asymptotic stability of systems of the form (\ref{eq:free})).
\suspend{listE}
\end{enumerate}
\addtocounter{listE}{1}
\begin{enumerate}[label= (A\arabic{listE}\alph*),leftmargin=*]
    \item
    \label{C:switchStable}
    Switching is exclusively dependent on 
    stable modes:
    \begin{equation}
    PV^{-1}=\begin{bmatrix}
    \decoup{P}^\stab  & 0_{n_P\times n_\unstab}
    \end{bmatrix}
    \end{equation}
    \item
    \label{C:switchUnstable}
    Switching is exclusively dependent on 
    unstable modes and 
    all unstable states 
    arising from (\ref{eq:inverse})
    are reachable in one time step from some predecessor state 
    $\forall k$:
    \begin{multline}
    \hspace{-0.15in}
    \left(
    PV^{-1}=\begin{bmatrix} 0_{n_P\times n_\stab} &  \decoup{P}^\unstab \end{bmatrix}
    \right)
    \,\,
    \\
    \hspace{-0.15in}
    \land
    \,\,
    \left(
    \forall k,\,\,\forall \statedc^\unstab\kdxplus{1}\in\mathcal{X}^\unstab\subseteq\real^{n_\unstab}  \,\,\, \text{Pre}(\{\statedc^\unstab\kdxplus{1}\})\neq \emptyset
    \right)
    \end{multline}
    where $\statedc^\unstab\kdx=
    \begin{bmatrix} 0_{n_\unstab\times n_\stab}& I_{n_\unstab\times n_\unstab} \end{bmatrix}
    V\state\kdx$.
    $\mathcal{X}^\unstab$ is 
    a set containing at least all solution values of $\statedc^\unstab\kdxplus{1}$ (see Section \ref{sec:practicalSI} for elaboration).
    $\text{Pre}(\bullet)$ is the set of predecessor states whose one-step successors belong to the set $\bullet$. 
\end{enumerate}

As implied by the separation of 
(A\arabic{listE})
into two opposing assumptions, the challenges associated with switching management precipitate different stable inversion procedures for the stable-mode-dependent 
switching 
and unstable-mode-dependent 
switching 
cases. In general, the trajectories of the modes 
upon which switching is dependent 
are computed first.
This allows the switching signal for the overall system to be computed and given as an exogenous input to the 
propagation
of the remaining modes,
ensuring all modes are correctly computed with the same location at each point in time.

The theorems for these cases are supported by the following notation for the decoupled system in addition to the above-defined 
$\Asumdc^\stab\kdx$, $\Asumdc^\unstab\kdx$, $\Pdc^\stab$, $\Pdc^\unstab$.
\begin{equation}
\begin{aligned}
    \statedc\kdx^\stab &\coloneqq \extracts V\state\kdx
    &
    \Bsumdc\kdx^\stab &\coloneqq \extracts V\inv{\Bsum}\kdx
    &
    \Fsumdc\kdx^\stab &\coloneqq \extracts V\inv{\Fsum}\kdx
    \\
    \statedc\kdx^\unstab&\coloneqq \extractu V\state\kdx
    &
    \Bsumdc\kdx^\unstab &\coloneqq \extractu V\inv{\Bsum}\kdx
    &
    \Fsumdc\kdx^\unstab &\coloneqq \extractu V\inv{\Fsum}\kdx
\end{aligned}
\end{equation}
\begin{equation}
    \extracts \coloneqq \begin{bmatrix} I_{n_\stab\times n_\stab} & 0_{n_\stab\times n_\unstab} \end{bmatrix}
    \qquad
    \extractu \coloneqq \begin{bmatrix} 0_{n_\unstab\times n_\stab}& I_{n_\unstab\times n_\unstab} \end{bmatrix}
\end{equation}
where
$I_{\bullet\times\bullet}$ is the square identity matrix of dimension $\bullet$.

\begin{theorem}[\PWA{} Stable Inversion with Stable-Mode-Dependent Switching]
\label{thm:SIstable}
Given an explicit inverse \PWA{} system (\ref{eq:inverse}) satisfying \ref{C:boundedSystem}-\ref{C:decouplable} and \ref{C:switchStable}, 
the solution to the stable inversion problem exists and can be found by
first computing the stable mode time series $\statedc^\stab\kdx$ and location time series $\locvec\kdx$ $\forall k$ forwards in time via
\begin{align}
    \locvec\kdx &= H\left( \decoup{P}^\stab\statedc^\stab\kdx  -\offsetvec \right)
    \label{eq:locstable}
    \\
    \statedc^\stab\kdxplus{1} &= \Asumdc^\stab\kdx\statedc^\stab\kdx + \Bsumdc^\stab\kdx y\kdxplus{\dmu} + \Fsumdc^\stab\kdx
    \label{eq:forwardstable}
\end{align}
The location time series being now known, the unstable mode time series $\statedc^\unstab\kdx$ can be computed backwards in time via
\begin{equation}
    \statedc^\unstab\kdx = \left(\Asumdc^\unstab\kdx\right)^{-1}\left(\statedc^\unstab\kdxplus{1}-\Bsumdc^\unstab\kdx y\kdxplus{\dmu} - \Fsumdc^\unstab\kdx\right)
    \label{eq:backwardstable}
\end{equation}
with $\locvec\kdx$ input directly to 
the selector functions in
(\ref{eq:Mshorthand}).
Finally the solution $u\kdx$ is computed via (\ref{eq:inverseOutput}) with
${\state\kdx=V^{-1}\left[\left(\statedc^\stab\kdx\right)^T,\,\,\left(\statedc^\unstab\kdx\right)^T\right]^T}$.
\end{theorem}
\begin{proof}
The prescribed formula
(\ref{eq:locstable})-(\ref{eq:backwardstable})
represents a solution to the 
boundary value problem of Definition \ref{def:SIproblem}
because
\begin{itemize}[leftmargin=*]
\item
by \ref{C:switchStable}, (\ref{eq:locstable}) is equivalent to (\ref{eq:localization}), 
\item
by \ref{C:decouplable}, the concatenated evolutions of (\ref{eq:forwardstable}) and (\ref{eq:backwardstable}) are equivalent to (\ref{eq:inverseState}), and
\item
by 
\ref{C:refdecay}
and the operator norm property $\norm{Mv}\leq\norm{M}\norm{v}$ for 
matrix $M\in\{\extracts,\extractu\}$ and vector $v=\inv{\Bsum}\kdx y\kdxplus{\dmu} + \inv{\Fsum}\kdx$,
(\ref{eq:forwardstable}) and (\ref{eq:backwardstable}) decay to the form of (\ref{eq:free}) in the limits as $k$ approaches $\infty$ or $-\infty$, and thus by \ref{C:decouplable} the boundary conditions at these limits are satisfied.
\end{itemize}
The
boundary value problem
solution
$\state\kdx$
and stable inversion solution $u\kdx$ are
guaranteed to exist because
\begin{itemize}[leftmargin=*]
\item
(\ref{eq:locstable})-(\ref{eq:backwardstable}) and (\ref{eq:inverseOutput}) are all explicit functions with all variables in the right-hand side known due to the order of time series computation, 
\item
boundedness of $\state\kdx$ as $k\rightarrow\pm\infty$ implies boundedness of $\state\kdx$ $\forall k$ by the facts that $\inv{\Asum}\kdx$, $\inv{\Bsum}\kdx$, $\inv{\Fsum}\kdx$ must be real (Definition \ref{def:PWAsys}, Theorems \ref{thm:mu0inv} and \ref{thm:genInv}),
and $\Asumdc^\unstab\kdx$ is guaranteed invertible by the eigenvalue condition of \ref{C:decouplable},
and
\item 
by \ref{C:boundedSystem} bounded of $\state\kdx$ and $y\kdxplus{\dmu}$ trajectories yield a bounded $u\kdx$ solution trajectory in (\ref{eq:inverseOutput}).
\end{itemize}
\end{proof}

\begin{theorem}[\PWA{} Stable Inversion with Unstable-Mode-Dependent Switching]
\label{thm:SIunstable}
Given an explicit inverse \PWA{} system (\ref{eq:inverse}) satisfying \ref{C:boundedSystem}-\ref{C:decouplable} and \ref{C:switchUnstable}, the solution to the stable inversion problem exists and can be found 
as follows.
First solve the implicit backward-in-time evolution of the unstable modes, (\ref{eq:backwardstable}), for $\statedc^\unstab\kdx$ at each time step using any of the applicable algorithms (e.g. brute force search over all locations or computational geometry methods 
\cite{Rakovic2006},
see Section \ref{sec:practicalSI} for elaboration).
For each time step at which one of the potentially multiple solution values of $\statedc^\unstab\kdx$ is chosen (any selection method is valid), the location vector $\locvec\kdx$ may be computed by
\begin{equation}
    \locvec\kdx=H\left(\Pdc^\unstab\statedc^\unstab\kdx-\offsetvec\right)
\end{equation}
The location time series being computed, $\locvec\kdx$ may be directly plugged in to
the selector functions in
(\ref{eq:Mshorthand}) to make the forward-in-time evolution of the stable modes, (\ref{eq:forwardstable}), explicit such that it can be evaluated at each time step for $\statedc^\stab\kdx$.
The solution $u\kdx$ is then computed via (\ref{eq:inverseOutput}) with 
$\state\kdx=V^{-1}\left[\left(\statedc^\stab\kdx\right)^T,\,\left(\statedc^\unstab\kdx\right)^T\right]^T$.
\end{theorem}
\begin{proof}
\begin{sloppypar}
The prescribed formula represents a solution to the stable inversion problem for the same reasons as Theorem \ref{thm:SIstable}, but with the first proposition of \ref{C:switchUnstable}---i.e. $PV^{-1}=
\left[
0_{n_P\times n_\stab} 
,\,\,
\decoup{P}^\unstab
\right]
$---used in place of \ref{C:switchStable}. The solution is guaranteed to exist because
\end{sloppypar}
\begin{itemize}[leftmargin=*]
\item
the second proposition of \ref{C:switchUnstable} guarantees the solution set of (\ref{eq:backwardstable}) is non-empty,
\item
the \PWA{} nature of the original inverse system (\ref{eq:inverse}) enables application of existing algorithms guaranteed to find $\text{Pre}(\{\statedc^\unstab\kdxplus{1}\})$ and thus solve (\ref{eq:backwardstable})
\cite{Rakovic2006},
and
\item
with solutions to (\ref{eq:backwardstable}) chosen $\forall k$, the remaining equations are explicit with bounded outputs for the same reasons as in Theorem \ref{thm:SIstable}.
\end{itemize}
\end{proof}

Note that while
\cite{Rakovic2006} focuses on \PWA{} systems with time-invariant components, 
for stable inversion
only the one-step predecessor set need be computed at each time step. 
Thus, because a time-varying system's propagation is equivalent to a sequence of different time-invariant systems' one-step propagations,
the algorithms of \cite{Rakovic2006} are still applicable.

\subsection{Practical Considerations}
\label{sec:practicalSI}

The most immediate issue with Theorems \ref{thm:SIstable} and \ref{thm:SIunstable} is that, while they provide an exact solution to the stable inversion problem, their procedures cannot be implemented because 
the time series involved are infinite.
This issue applies to past works on stable inversion as well, and the same means of addressing the issue is taken here.
Namely, an approximate solution is 
obtained
by prescribing a finite reference $\refr\kdxplus{\dmu}$ for $k\in\llbracket 0,\,\,N-\dmu \rrbracket$ and strictly enforcing the boundary conditions on only the 
initial and terminal states of the stable and unstable mode evolutions, respectively.
In other words, $\statedc^\stab_{0}=0$ and $\statedc^\unstab_{N-\dmu}=0$ but $\statedc^\stab_{N-\dmu}$ and $\statedc^\unstab_{0}$ may be nonzero.

In
general the closer $\statedc^\stab\kdx$ and $\statedc^\unstab\kdx$ come to decaying to zero by $k=N-\dmu$ and $k=0$, respectively, the higher quality the approximation of the $u\kdx$ time series will be, i.e. the closer $u\kdx$ comes to returning $y\kdxplus{\dmu}=\refr\kdxplus{\dmu}$ when input to the original system from which the inverse (\ref{eq:inverse}) was derived.
To achieve this high quality approximation, one may specify $\refr\kdxplus{\dmu}$ to begin and end with a number of 
``force-zeroing''
elements to allow space for the $u\kdx$ time series to contain the pre- and post-actuation typically necessary for the control of NMP systems.
``Force-zeroing'' elements are those with values such that the forcing part of the inverse's state function (\ref{eq:inverseState}) is zero.
In other words, leading and trailing force-zeroing reference elements support approximate satisfaction of \ref{C:refdecay}, which 
supports approximate satisfaction of the $\state\kdx$ boundary conditions, which reduces 
inversion error.
The number of 
force-zeroing
elements required to achieve a satisfactorily low error is case-dependent. One typical heuristic is to ensure that the durations of the leading and trailing 
force zeroing are
approximately equal to the system settling time.

In addition to 
\ref{C:refdecay} and 
the 
practical 
need for
finite references with 
force zeroing, the case of unstable-mode-dependent switching and \ref{C:switchUnstable} 
merits
special 
attention 
regarding
implementation.

First, consider methods to solve
the implicit equation
(\ref{eq:backwardstable}) at each time step. Any method will consist of two parts: identifying a set of valid solutions and then choosing one of them.
Choosing a solution can be formalized as the minimization of some cost function, 
$\norm{\statedc^\unstab\kdxplus{1}-\statedc^\unstab\kdx}$ being a straightforward and universally applicable option. If the inverse system (\ref{eq:inverse}) has exclusively unstable modes, input-based costs such as $\norm{u\kdx}$ and $\norm{u\kdxplus{1}-u\kdx}$ may also be used.
All these optimizations are
combinatorial, i.e. 
the decision variable can only take on values from a particular finite set.
For the problem considered here, 
the cardinality of this 
set 
is at most
the number of locations $\regionQuant$ in the system.
This is because each location contains at most one solution to (\ref{eq:backwardstable}) due to the eigenvalue condition in \ref{C:decouplable} making $(\Asumdc^\unstab\kdx)^{-1}$ full rank and thus one-to-one.

Having an upper bound of one solution per location leads to a direct method for deriving the set of valid solutions to (\ref{eq:backwardstable}). For each location $\Qset_q\in\Qset$, compute
\begin{equation}
    \left(\decoup{A}^\unstab_{q,k}\right)^{-1}
    \left(\statedc^\unstab\kdxplus{1}-\decoup{B}_{q,k}^\unstab \refr\kdxplus{\dmu}-\decoup{F}_{q,k}^\unstab\right)
\end{equation}
and check whether the result lies in $\Qset_q$. If so, the result is a solution to (\ref{eq:backwardstable}). 
Logic
relating to solution selection criteria may be incorporated to reduce computational cost, e.g. checking locations in order of proximity to the current location to avoid checking all locations in the case that the solution selection cost function is something like $\norm{\statedc^\unstab\kdxplus{1}-\statedc^\unstab\kdx}$.
Alternatively, $\text{Pre}\left(\left\{ \statedc^\unstab\kdxplus{1} \right\}\right)$ may be derived whole using the computational-geometry-supported algorithms of 
\cite{Rakovic2006}.
As noted in 
\cite{Rakovic2006},
the algorithm of least cost may be case-dependent.

Finally, consider verification of the existence of a solution to the stable inversion problem with unstable-mode-based switching.
The verification method recommended here is to first verify \ref{C:boundedSystem}-\ref{C:decouplable} and the first proposition of \ref{C:switchUnstable} directly, then run the procedure given in Theorem \ref{thm:SIunstable} for finding a solution.
If \ref{C:boundedSystem}-\ref{C:decouplable} and the first proposition of \ref{C:switchUnstable} have been verified, then the procedure is guaranteed to find a solution to (\ref{eq:backwardstable}) if a solution exists. Equivalently, failure to find a solution implies no solution exists.

This method is recommended over directly attempting to verify the second proposition of \ref{C:switchUnstable} because this 
proposition may be conservative, difficult to verify, and yield 
limited computational savings over the recommended method.
Conservativeness
and difficulty arise
from 
choosing the
set of possible $\statedc^\unstab\kdxplus{1}$ values, $\mathcal{X}^\unstab$.
It is unlikely 
for one
to 
have
knowledge of the 
solution
$\statedc^\unstab\kdxplus{1}$ values 
prior to actually solving the stable inversion 
problem 
(via Theorem \ref{thm:SIunstable} and a method for solving (\ref{eq:backwardstable}) as described above),
so $\mathcal{X}^\unstab$ may need to be set much larger than necessary to ensure it contains the solution trajectory.
Containment
is necessary for truth of the proposition to imply 
solution existence verification.
Conversely, the proposition may evaluate to false despite containing a true solution if $\mathcal{X}^\unstab$ also contains unreachable states.
Selection of $\mathcal{X}^\unstab$ may thus be 
delicate and challenging.

The expectation of high computational cost arises from the subtle discrepancy between the capabilities of established \PWA{} system verification methods and the second proposition of \ref{C:switchUnstable}.
Multiple methods exist for verifying the reachability/controllability of a \PWA{} system to a target set $\mathcal{X}^\unstab$
\cite{Rakovic2006,Bemporad2000}.
However, these methods typically verify whether 
at least one
element of the target set is reachable, whereas 
\ref{C:switchUnstable} requires
verification that every element of the target set is reachable.
Methods for complete reachable set computation (e.g. \cite[ch. 6]{Torrisi2003})  typically over-approximate the reachable set because this is conservative in safety verification, but it is anti-conservative in the present context.
In other words, the computational savings  
one might expect from
existing reachability/controllability verification schemes may not be 
available
for formal verification.
However, if ``informal'' verification is acceptable one may guess a set containing the solution values of $\statedc^\unstab\kdxplus{1}$ and check if the system's (potentially approximate) one-step reach set contains this $\mathcal{X}^\unstab_\textrm{guess}$ via the techniques of \cite[ch. 6]{Torrisi2003}. While guarantees of correctness are lost, such approximate verification may be useful if one has advance knowledge of what solution values are likely to be.

In short, 
the 
second
proposition
of \ref{C:switchUnstable} is useful for specifying a condition 
guaranteeing solution existence,
and thus for 
derivation and proof of Theorem \ref{thm:SIunstable}. 
It is \emph{not} recommended as a tool for existence verification, but this is not a significant loss because the procedure given in Theorem \ref{thm:SIunstable} for finding a solution is itself a valid verification tool.

Finally, consider the remaining stable inversion assumptions \ref{C:boundedSystem} and \ref{C:decouplable}.
Assumption \ref{C:boundedSystem} prevents unbounded growth from ``hiding'' in the time-varying matrices of the output function (\ref{eq:inverseOutput}), ensuring bounded $\state\kdx$ and $y\kdxplus{\dmu}$ implies bounded $u\kdx$.
This is easy to satisfy in practice because the time-variation allowed by Definition \ref{def:PWAsys} is arbitrary and the engineer has ultimate discretion to design the model (\ref{eq:PWAdef}). For example, if a $\dmu=0$ model has infinite growth of $\norm{\inv{\Dsum}\kdx}$ due to asymptotic decay of $\norm{D_{q,k}}$, a saturation-like limit to $D_{q,k}$ decay can be added at extremities of time to yield a nearly equivalent model satisfying \ref{C:boundedSystem}.
Assumption \ref{C:decouplable} ensures the stable and unstable modes can be propagated independently, but applies the restriction that a constant similarity transform must be able to perform the decoupling for all locations.
Finding a variable similarity transform to relax this constraint may be nontrivial. However, for many models based on physical systems the location dynamics are related, and the similarity transform yielding the Jordan form of one location's state matrix decouples the other locations' modes.

\section{Validation: Application to ILC}
\label{sec:val}

This section simulates the application of \ILILC{} and feedforward control to a \PWA{} system with \NMP{} components via the stable inversion theory of contribution \ref{contribution:SI} (Section \ref{sec:stabinv}), and thus also the conventional inversion theory of \ref{contribution:inversion}-\ref{contribution:uniqueness} (Section \ref{sec:exactinverse}).
Feedforward control alone 
uses
\ref{contribution:inversion}-\ref{contribution:SI}, and the
integration of \ILILC{} with \PWA{} stable inversion, \ref{contribution:sim}, demonstrates the performance improvement enabled by \ref{contribution:inversion}-\ref{contribution:SI} over alternative control techniques.

The validation scenario is an inkjet printhead positioning system using feedback and feedforward control simultaneously.
The complete list of simulations performed is
\begin{enumerate*}[label=(\alph*)]
\item
\label{sim:ILILC}
\ILILC{} with \PWA{} stable inversion,
\item
\label{sim:SI}
learning-free \PWA{} stable inversion (i.e. stable inversion without \ILILC{}),
\item
\label{sim:FB}
feedback-only 
Proportional-Derivative (PD)
control (i.e. zero feedforward input),
\item
gradient ILC, and
\item
\label{sim:P}
P-type ILC.
\end{enumerate*}
Simulations \ref{sim:ILILC}-\ref{sim:SI} represent this article's contribution; \ref{sim:FB}-\ref{sim:P} are benchmarks.
Gradient ILC is an optimization-based ILC framework that has been used to benchmark \ILILC{} in past works \cite{Spiegel2021}.
P-type ILC is used as 
an additional benchmark
because 
it
is among the most common forms of ILC used in industry,
and
is considered by many to be a form of ``model-free'' ILC.
Synthesis details for
all ILC schemes are given in Section \ref{sec:ILCSchemes}.

These simulations are
subject to a variety of model errors and other disturbances to validate
the stable-inversion-supported learning controller's practicality.
In other words, the controller is synthesized from a ``control model'' and applied to a ``truth model'' representing a physical system. The control model features mismatches in parameter values, sample rate, model order, and relative degree.
This model mismatch captures disturbances that can arise in the real operation of a printhead positioning system. For example, calibration errors in low-level motor control can result in flawed scaling of the voltage delivered to the 
motor; this
disturbance is captured by the ratio between the truth and control model gain magnitudes.

Additionally, while the physical system has virtually no output-measurable noise, 
this simulation
is subject to copious process and measurement noise. 
This is done as a preliminary test to ensure that noise does not corrupt the learning process beyond the remedial power of conventional filtering.

\subsection{Example System}
The truth model is based on 
a
physical desktop inkjet printhead positioning 
system,
pictured in Figure \ref{fig:photo}.
The input to this system is an applied motor voltage, $\controltotal\kdx$, and the output is the printhead position along a \SI{0.3}{\meter} guide rail, $y^{P}\kdx$, measured by a linear optical encoder with a resolution of \SI{50}{\micro\meter}.
The applied motor voltage $\controltotal\kdx$ is the sum of a feedback component, $\controlfb^{C}\kdx$, and feedforward component, $u\kdx$. 
Gaussian white process noise $\omega_{\controltotal}$ (zero mean, standard deviation \SI{0.03}{\volt}) and measurement noise $\omega_{y}$ (zero mean, standard deviation \SI{50}{\micro\meter}) are added at the input and output of the plant, respectively. This ultimately results in the block diagram of Figure \ref{fig:blockA3}.

\begin{figure}
    \centering
    \includegraphics[trim={0.5mm 0.2mm 0mm 0.2mm},clip]{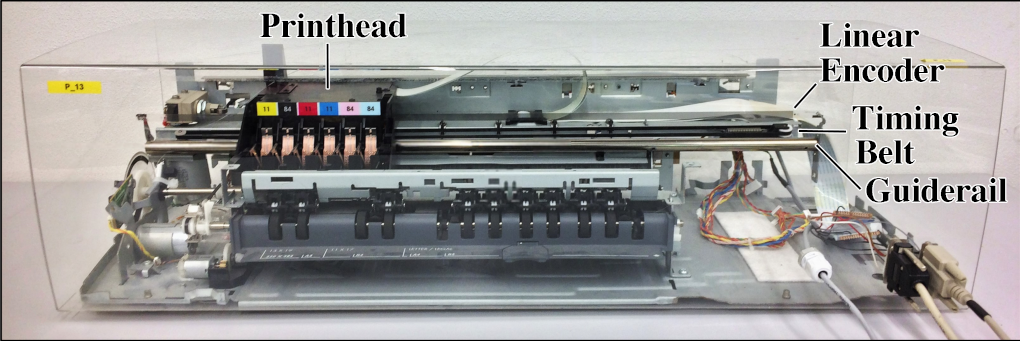}
    \caption{
    Photo of desktop inkjet printer with the case removed.
    The motor actuates the printhead motion along a guide rail via a timing belt, and the motion is measured by a linear optical encoder with resolution of \SI{50}{\micro\meter} (about 600 dots per inch).}
    \label{fig:photo}
\end{figure}

\begin{figure}
    \centering
    \includegraphics[width=\columnwidth]{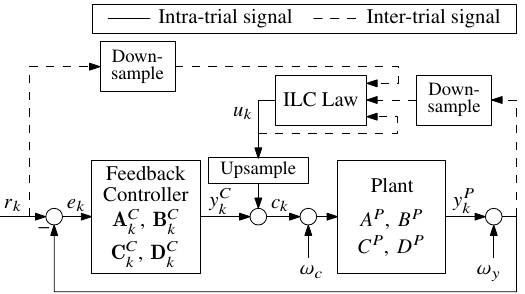}
    \caption{
    System block diagram. 
    The plant block uses the truth model of the printer system obtained by experimental system identification 
    while the 
    ILC law is synthesized using the control model. The downsample and upsample blocks account for the difference in sample period between the ILC law and the truth model.
    }
    \label{fig:blockA3}
\end{figure}

The plant models are Linear Time-Invariant (LTI), with the truth model being identified from historical experimental data.
The feedback controller is composed of a lowpass filter in series with a switching PD controller. 
All locations yield a DC Gain of 1 between the reference and the output, so integral control is not 
necessary from the typical steady-state-error perspective. The control model of this feedback controller has the state space representation
\begin{IEEEeqnarray}{RL}
\nonumber
    \AcmC &= \begin{bmatrix}
    0 & 1 & 0 \\
    -a_2 & -a_1 & 0 \\
    0 & 0 & 0
    \end{bmatrix}
\quad\,\,\,\,\,
    \BcmC = \begin{bmatrix}
    0 \\ 1 \\ 1
    \end{bmatrix}
\quad\,\,\,\,\,
    \DcmC=b\left(K_p+\frac{K_d}{T_s}\right)
    \\
\label{eq:feedbackSS}
    \CcmC&= 
    -b\begin{bmatrix}
    \frac{K_d(1+a_2)}{T_s} + K_pa_2
    &
    \frac{K_da_1}{T_s} + K_p(a_1-1)
    &
    0
    \end{bmatrix}
    \\
\nonumber
    K_p 
    &= 
    \begin{cases}
    K_{p,1} 
    &
    \left|e\kdxminus{1}\right| \leq e_{\text{switch}}
    \\
    K_{p,2} 
    &
    \left|e\kdxminus{1}\right| > e_{\text{switch}}
    \end{cases}
    \qquad
    e\kdxminus{1} = \begin{bmatrix} 0 & 0 & 1 \end{bmatrix}\xcm^C\kdx
\end{IEEEeqnarray}
where $\xcm^{C}\kdx\in\real^{n_{\xcm^{C}}}$ is the state vector, $a_1$, $a_2$, $b\in\real$ are lowpass filter parameters, $T_s\in\real_{>0}$ is the sample period in seconds, $e\kdx$ is the reference-output error, $e_{\text{switch}}$ is an error magnitude switching threshold, and $K_p$, $K_d\in\real_{>0}$ are the proportional and derivative gains. Hats, $\hat{}$, are used to distinguish control model matrices, state vectors, and outputs from those of the truth model.
For both the control and truth models, Table \ref{table:modelParams} specifies the plant in terms of pole, zero, and gain values, and specifies the controller in terms of its lowpass filter and PD control parameters.
See Appendix \ref{sec:A1} for further details.

\begin{table}
    \centering
    \caption{Simulation Model Parameters}
    \label{table:modelParams}
    \setstretch{1.1}
    \newcommand{\pzhrule}{\hhline{|>{\arrayrulecolor{lightgray}}->{\arrayrulecolor{black}}||>{\arrayrulecolor{lightgray}}->{\arrayrulecolor{black}}|>{\arrayrulecolor{lightgray}}->{\arrayrulecolor{black}}|}}
    \begin{tabular}{|c||c|c|}
\hline
            & Truth Model         & Control Model       \\ \hhline{|=#=|=|}
Plant Poles & $0.88\pm 0.37 i$    & $0.67 \pm 0.61 i$   \\ 
            \pzhrule
            & $1.00$              & $0.99$              \\ 
            \pzhrule
            & $1.00$              & $1.00$              \\ 
            \pzhrule
            & $0$                 & N/A                 \\ \hline
Plant Zeros & $-5.10$             & $33.10$             \\ 
            \pzhrule
            & $-0.44$             & $-2.21$             \\ 
            \pzhrule
            & $\phantom{-}0.16$              & $\phantom{-}0.16$              \\ \hline
Plant Gain  & \SI{2.42e-7}{}      & \SI{-2.38e-7}{}      \\ \hline
$a_1$       & $-1.65\phantom{0}$             & $-1.31\phantom{0}$             \\ \hline
$a_2$       & $\phantom{-}0.70\phantom{0}$              & $\phantom{-}0.50\phantom{0}$              \\ \hline
$b$         & $\phantom{-}0.027$             & $\phantom{-}0.093$             \\ \hline
$K_d$       & $3$                 & $3$                 \\ \hline
$K_{p,1}$   & $40$                & $40$                \\ \hline
$K_{p,2}$   & $160$               & $160$               \\ \hline
$e_{\text{switch}}$ & \SI{2}{\milli\meter} & \SI{2}{\milli\meter} \\ \hline
$T_s$       & \SI{0.001}{\second} & \SI{0.002}{\second} \\ \hline
\end{tabular}
\end{table}

To perform stable inversion of the 
dynamics from 
feedforward input to 
output,
a monolithic \PWA{} model
\sysDef{} encompassing the feedback controller \emph{and} plant is needed. That is,
\begin{IEEEeqnarray}{RL}
\label{eq:monox}
\IEEEyessubnumber
    \statemod\kdxplus{1} &= \Asummod\kdx\statemod\kdx + \Bsummod\kdx u\kdx + \Fsummod\kdx
    \\
    \label{eq:monoy}
\IEEEyessubnumber
    \ymod\kdx^P &= \Csummod\kdx\statemod\kdx + \Dsummod\kdx u\kdx + \Gsummod\kdx
\end{IEEEeqnarray}%
\begin{equation}
\begin{aligned}
    \Asummod\kdx &= 
    \begin{bmatrix}[1.375]
    \Acm^{P} - \Bcm^{P}\DcmC\Ccm^{P}
    &
    \Bcm^{P}\CcmC
    \\
    -\BcmC\Ccm^{P}
    &
    \AcmC
    \end{bmatrix}
    \qquad
    \Bsummod\kdx = 
    \begin{bmatrix}[1.375]
    \Bcm^{P} \\ 0_{n_{\xcm^{C}} \times 1}
    \end{bmatrix}
    \\
    \Csummod\kdx &=
    \begin{bmatrix}
    \Ccm^{P} & 0_{1\times n_{\xcm^{C}}}
    \end{bmatrix}
    \qquad
    \Dsummod\kdx = 0
    \\
    \Fsummod\kdx &=
    \begin{bmatrix}[1.375]
    \Bcm^{P}\DcmC \\ \BcmC
    \end{bmatrix}
    \refr\kdx
    \qquad \quad \,\,
    \Gsummod\kdx=0
    \qquad
    \statemod\kdx = 
    \begin{bmatrix}[1.375]
    {\xcm^{P}\kdx} \\ {\xcm^{C}\kdx}
    \end{bmatrix}
\end{aligned}
\end{equation}
where ($\Acm^{P}$, $\Bcm^{P}$, $\Ccm^{P}$, $\Dcm^{P}$) 
is the state space representation of the plant's control model,
with state vector $\xcm^{P}\kdx\in\real^{n_{\xcm^{P}}}$.
It is assumed that the feedforward input $u_k$ is equivalent between the truth and control model, so it is not hatted.
This system is \NMP{}, as it has all the zeros of the plant model given in Table\nolinebreak{} \ref{table:modelParams} (as well as additional zeros).

This
monolithic control model has two locations based on the switching of $K_p$ in $\CcmC$ and $\DcmC$. Let the location $q=1$ correspond to low error with $K_{p,1}$ and $q=2$ correspond to high error and $K_{p,2}$. Then the switching parameters are
\begin{equation}
\begin{aligned}
    \hat{P} &= \begin{bmatrix}
    0_{1\times n_{\xcm^{P}}+n_{\xcm^{C}}-1} & -1
    \\
    0_{1\times n_{\xcm^{P}}+n_{\xcm^{C}}-1} & \phantom{-}1
    \end{bmatrix}
    \qquad
    \hat{\offsetvec} = \begin{bmatrix}
    -e_{\text{switch}} \\ -e_{\text{switch}}
    \end{bmatrix}
    \\
    \sigvecsetmod_1 &= \left\{ \begin{bmatrix} 1\\1 \end{bmatrix} \right\}
    \qquad
    \sigvecsetmod_2 = \left\{ \, \begin{bmatrix} 1\\0 \end{bmatrix},\,\begin{bmatrix}0\\1\end{bmatrix}\, \right\}
\end{aligned}
    \label{eq:deltaCellCell}
\end{equation}
Note that $\locvec\kdx = [ 0, \,\, 0]^T$ is not reachable, and thus need not be included in 
$\sigvecsetmod_2$.

This monolithic model is of global dynamical relative degree $\dmu=1$ and satisfies \ref{C:reachable}-\ref{C:locIndep}, enabling the use of Corollary \ref{corollary:mu1} for derivation of the conventional inverse.
The resultant inverse system satisfies
\ref{C:boundedSystem}-\ref{C:switchStable},
enabling the use of 
stable inversion
for the generation of stable
inverse state trajectories.
The decoupling similarity transform $V$ is derived by 
MATLAB's \texttt{canon} applied to location 1's dynamics.

Finally, the reference $r\kdx$ is given in Figure \ref{fig:referenceA3}.

\begin{figure}
    \centering
    \includegraphics[scale=0.8]{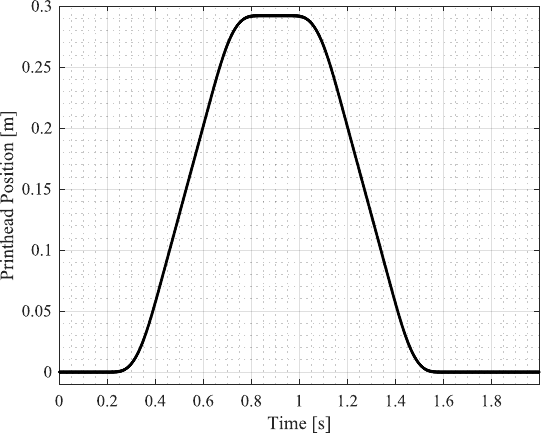}
    \caption{Reference. The reference is 1999 samples long for the truth model, and is downsampled to 1000 samples for the control model.}
    \label{fig:referenceA3}
\end{figure}

\subsection{ILC Schemes}
\label{sec:ILCSchemes}
Each ILC scheme in this article represents a different method of deriving the (potentially) trial-varying learning matrix $\learnMat_\tdx\in\real^{N-\dmu+1\times N-\dmu+1}$ of the ILC law
\begin{equation}
    \uLift_{\ldx+1} = \mathscr{E}\mathscr{Q}\left(\uLift_\ldx + \learnMat_\ldx\left(\rLift-\mathscr{Q}\yLift_\ldx\right)\right)
    \label{eq:ILCfiltered}
\end{equation}
where $\ell\in\integer_{\geq 0}$ is the iteration/trial index, $\uLift$, $\yLift$, $\rLift\in\real^{N-\dmu+1}$ are the lifted vectors (i.e. time series vectors)
\begin{align}
    \uLift
    &= 
    \begin{bmatrix}
    u_{0} & u_{1} & \cdots & u_{N-\dmu}
    \end{bmatrix}^T
    \label{eq:uLift}
    \\
    \yLift
    &=
    \begin{bmatrix}
    y_{\dmu} & y_{\dmu+1} & \cdots & y_{N}
    \end{bmatrix}^T
    \label{eq:yLift}
    \\
    \rLift &= \begin{bmatrix}
    r_{\dmu} & r_{\dmu+1} & \cdots & r_{N}
    \end{bmatrix}^T
\end{align}
and $N\in\integer_{>\dmu}$ is the number of time steps in a reference tracking trial (the number of samples is $N+1$).
Reference $\rLift$ is trial-invariant.
$\mathscr{Q}$, $\mathscr{E}\in\real^{N-\dmu+1\times N-\dmu+1}$ are 
lowpass and edge-effect filters, respectively (see Appendix \ref{sec:A2} for details). For the printhead system, $\uLift_0=0_{N-\dmu+1\times 1}$ and $y\kdx=y^{P}\kdx+\omega_y$.

\subsubsection{Invert-Linearize ILC}

\ILILC{} is
introduced in \cite{Spiegel2021} as
the application 
to ILC 
of the Newton-Raphson root finding algorithm with 
structure modifications enabling
numerical conditioning improvements. 
To derive $\learnMat_\ldx$, \ILILC{} calls for a lifted input-output model inverse $\ginvLift:\real^{N-\dmu+1}\rightarrow\real^{N-\dmu+1}$ taking in the \emph{measured} output $\yLift$ and outputting the control signal $\uLift$ \emph{predicted} to yield $\yLift$ when input to the true, unknown system.
Equivalently, $\ginvLift$ takes in the \emph{model} output $\yLiftmod$ and outputs the control signal $\uLift$ that yields $\yLiftmod$ when input to the known model \emph{approximating} the true system.

This $\ginvLift$ must be 
closed-form, such that the \ILILC{} learning matrix
\begin{equation}
    \learnMat_\ldx = \jacobian{\ginvLift}{\yLiftmod}\left(\yLift_\ldx\right)
    \label{eq:learnMat_ILILC}
\end{equation}
can be derived via an automatic differentiation tool such as CasADi \cite{Andersson2018}. Here, $\jacobian{\ginvLift}{\yLiftmod}$ is the Jacobian (in numerator layout) of $\ginvLift$ with respect to $\yLiftmod$,
as a function of output.
Furthermore, 
when the known system model is \NMP, $\jacobian{\ginvLift}{\yLiftmod}$ is likely to be ill-conditioned unless $\ginvLift$ is synthesized using stable inverse 
trajectories \cite{Spiegel2021}.
For a \PWA{} system satisfying \ref{C:boundedSystem}-\ref{C:switchStable}, Theorem \ref{thm:SIstable} provides a method for generating these closed-form state trajectories;
given 
known
trial-invariant conditions 
$\decoup{\state}_{0}^{\stab}=0$ and $\decoup{\state}_{N-\dmu}^{\unstab}=0$,
each element of the time series $\statemod\kdx$ and $\locvecmod\kdx$ is a function only of $\yLiftmod$. Then $\ginvLift$ is given by
\begin{equation}
    \ginvLift = \begin{bmatrix}
    \Csummod_{0}\statemod_{0} + \Dsummod_{0}\ymod_{\dmu} + \Gsummod_0
    \\
    \Csummod_{1}\statemod_{1} + \Dsummod_{1}\ymod_{1+\dmu} + \Gsummod_1
    \\
    \vdots
    \\
    \Csummod_{N-\dmu}\statemod_{N-\dmu} + \Dsummod_{N-\dmu}\ymod_{N} + \Gsummod_{N-\dmu}
    \end{bmatrix}
\end{equation}

\subsubsection{Gradient and Lifted P-Type ILC}

Gradient ILC uses
\begin{equation}
    \learnMat_\tdx = \gradILCgain \jacobian{\gmod}{\uLift}(\uLift_\tdx)^T
    \label{eq:learnMat_grad}
\end{equation}
where $\gradILCgain\in\real_{>0}$ is the learning gain and $\gmod:\real^{N-\dmu+1}\rightarrow\real^{N-\dmu+1}$ is the lifted system input-output model
\begin{equation}
\yvecmod = \gmod(\uLift)
\end{equation}
which can be synthesized from a 
model of the form (\ref{eq:sysDefPWA}) as
\begin{IEEEeqnarray}{RL}
\eqlabel{eq:g}
\IEEEyesnumber
\IEEEyessubnumber
\label{eq:gelementdef}
\yvecmod^i &= \gmod^i\left(\uLift\right)=\ymod_{\dmu+i-1}
\\
\ymod\kdx &= 
\Csummod\kdx
\left(\left( \prod^{k-1}_{m=0}\Asummod_m\right)\statemod_0 
\right.
\IEEEyessubnumber\label{eq:ynow}
\\
\nonumber
& \qquad
\left.
+
\sum^{k-1}_{s=0}\left( \left( \prod^{k-1}_{m=s+1}\Asummod_m\right)\left(\Fsummod_s + \Bsummod_s u_s \right)\right)\right)
\Dsummod\kdx u\kdx + \Gsummod\kdx
\end{IEEEeqnarray}
where the superscript $i$ in (\ref{eq:gelementdef}) denotes the vector element, indexing from 1, and $\dmu+i-1$ can be directly substituted for $k$ in (\ref{eq:ynow}).
Because the initial condition $\statemod_0$ is assumed known and trial-invariant, 
(\ref{eq:g}) 
is
a function only of $\uLift_\tdx$.

P-type ILC as described in \cite{Ratcliffe2005} manifests in lifted form as the trial invariant learning matrix
\begin{equation}
    \learnMat_\tdx = \PILCgain I_{N-\dmu+1 \times N-\dmu+1}
    \qquad
    \forall \tdx
\end{equation}
where 
$\PILCgain$ is a constant scalar.
This $\learnMat_\tdx$ 
can be plugged into the filtered learning law (\ref{eq:ILCfiltered}).
At this time, there is no literature prescribing a stable synthesis procedure for the P-type ILC of \PWA{} systems. In fact, the literature lacks application of either gradient ILC or P-type ILC to \PWA{} models.
However, because of its simplicity it is ubiquitous in industry and often synthesized heuristically, much like the PID feedback control for which it was named. Thus, P-type ILC makes for an important benchmark.

See Appendix \ref{sec:A2} for $\gradILCgain$ and $\PILCgain$ tuning.

\subsection{Assessment Methods}
\label{sec:methods5}
The presented stable inversion theory's ability to derive the inverse of a non-minimum phase \PWA{} model is tested by applying $\uLift=\ginvLift(\rLift)$ to the control model.
Then, to assess a more practical utility, 
five
independent simulations with the truth model are performed.
First,
a simulation is run with the feedforward input fixed to zero for all time, yielding a feedback-only 
PD control
simulation to serve as a baseline against which the four feedforward controllers can be compared.
Next,
learning-free stable inversion is applied to the truth model.
In other words, a simulation is run with $\uLift=\ginvLift(\rLift)$. 
The remaining three simulations each use one of the ILC techniques described in Section \ref{sec:ILCSchemes} (\ILILC{}, gradient ILC, or P-type ILC) with 9 trials (8 learning operations).

The primary metric for assessing control performance in a given trial is the normalized root mean square error (NRMSE) of the truth model, defined as
\begin{equation}
    \text{NRMSE} = 
    \frac{\RMS_{k\in\llbracket\dmu,N\rrbracket}\left(e\kdx\right)}{\max_{k\in\llbracket\dmu,N\rrbracket}\left(\left| r\kdx\right|\right)}
    =
    \frac{1}{\norm{\rLift}_\infty}
    \frac{\norm{\rLift-\yLift}_2}{\sqrt{N-\dmu+1}}
\end{equation}
The peak error magnitude
$\norm{\rLift-\yLift}_\infty$
is also considered.

\subsection{Results and Discussion}

When $\uLift=\ginvLift(\rLift)$ is input to the noise-free control model \controlModel{} from which it was derived, the resulting NRMSE and peak error magnitude are 
$10^{-7}$
and
\SI{49}{\nano\meter}. 
This is nearly zero compared to the other simulation errors, tabulated in Table  \ref{table:truthsim}, and what error there is can be attributed to the approximation error expected of finite-time stable inversion, as discussed in Section \ref{sec:practicalSI}.
This validates the fundamental theoretical contributions of this article: the stable inversion---and thus also conventional inversion---of \PWA{} systems.

To analyze the more practical application of these techniques to the truth model, Figure \ref{fig:NRMSEA3} plots the
ILC scheme NRMSEs evolving over the iteration process,
and
Figure \ref{fig:timeseriesA3combined}
plots
the input and error time series for the five simulations. 
The NRMSE and peak error magnitude for these simulations are tabulated in Table \ref{table:truthsim}.

\begin{table}
\caption{Simulation NRMSE and Peak Error Magnitude (PEM)}
    \label{table:truthsim}
    \centering
    \setstretch{1.2}
    \begin{tabular}{|lcc|}
\hline
Simulation Type                & NRMSE         & PEM                     \\ \hhline{|===|}
Feedback-Only 
(PD Control)
                               & \SI{5.6e-3}{} & \SI{3.9}{\milli\meter}\phantom{0} \\ \hline
Learning-Free Stable Inversion & \SI{5.4e-3}{} & \SI{2.8}{\milli\meter}\phantom{0} \\ \hline
P-type ILC - Final Trial       & \SI{3.1e-3}{} & \SI{2.1}{\milli\meter}\phantom{0} \\ \hline
Gradient ILC - Final Trial     & \SI{1.2e-3}{} & \SI{1.1}{\milli\meter}\phantom{0} \\ \hline
\ILILC{} - Final Trial            & \SI{3.6e-4}{} & \SI{0.34}{\milli\meter} \\ \hline
\end{tabular}
\end{table}


\begin{figure}
    \centering
    \includegraphics[width=\columnwidth]{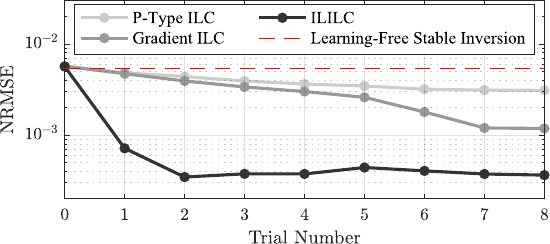}
    \caption{
    NRMSE of each ILC trial, illustrating convergence of \ILILC{} to a plateau determined by the noise injected to the system, and dramatically surpassing the convergence speed of both benchmark ILC techniques. The NRMSE of the learning-free stable inversion simulation is also pictured. It is 4\% smaller than the feedback-only simulation (\ILILC{} trial 0), but is much larger than the performance achievable with learning.
    }
    \label{fig:NRMSEA3}
\end{figure}

\begin{figure}
    \centering
    \includegraphics[width=\columnwidth]{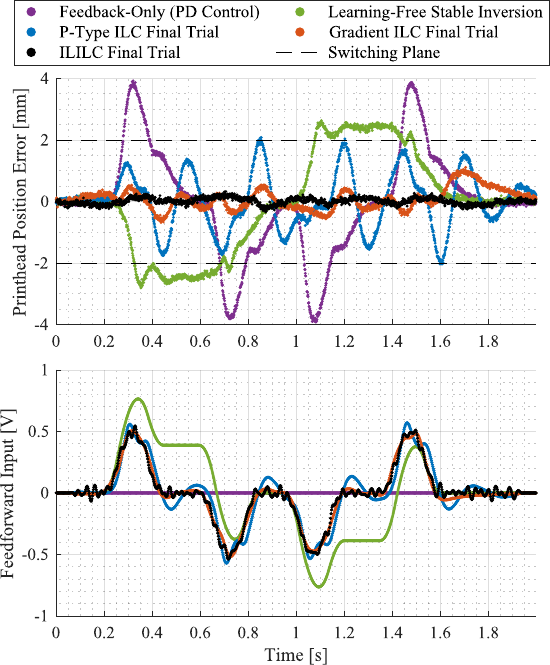}
    \caption{
    Error (top) and Input (bottom) time series data for the
    five simulations.
    Learning-free stable inversion performs as expected and yields benefits over feedback-only PD Control,
    but due to model error learning is required to reap the full benefit of feedforward control. \ILILC{} with stable inversion clearly yields the lowest-error performance. 
    This superiority is in spite of \ILILC{} acquiring more high frequency content via learning than the other ILC schemes, which appear less noisy but appear to contain higher amplitude, lower frequency oscillations that degrade performance.
    }
    \label{fig:timeseriesA3combined}
\end{figure}

The feedback-only PD controller has the worst performance. This is because PID control can only 
react
to preexisting errors 
while feedforward control can proactively prevent errors.
The model error is small enough that
learning-free stable inversion does yield some improvement 
(4\%) 
over the feedback-only NRMSE. 
(One
expects that the NRMSE of learning-free stable inversion would grow if the model error increased). Learning-free stable inversion also reduces the peak error magnitude, which is a critical safety criterion in many applications, by 
28\%.
However, because of the model error that does exist, all of the ILC schemes defeat learning-free stable inversion in both metrics. P-type ILC shows the least improvement, but still yields a 40\% reduction in NRMSE from learning-free stable inversion. This improvement is dwarfed by that of \ILILC{}, however, which yields a 70\% NRMSE improvement from gradient ILC (the next best technique), or equivalently, a 93\% improvement from learning-free stable inversion.

\ILILC's improvement over the 
benchmarks
is intuitive in light of the foundational theories of Newton's method (on which \ILILC{} is based) and gradient descent. 
Newton-Raphson root finding 
can achieve quadratic convergence while gradient descent is limited to a form of linear convergence.
More concretely, \cite{Avrachenkov1998} proves an upper bound for the convergence rate of the general ILC law 
$\uLift_{\tdx+1}=\uLift_\tdx+\learnMat_\tdx(\rLift-\yLift_\tdx)$
and that this bound decreases as the learning matrix $\learnMat_\ldx$ approaches the inverse of the true, unknowable system dynamics' Jacobian.
(Convergence rate is a ratio of ``next error'' over ``present error'', so the smaller the convergence rate, the faster the convergence).
\ILILC{} intentionally seeks this inversion via (\ref{eq:learnMat_ILILC}), while gradient ILC merely uses the model's Jacobian transpose in (\ref{eq:learnMat_grad}), which does \emph{not} approximate the inverse.
Figure \ref{fig:NRMSEA3} illustrates the resultant advantage in convergence rate of \ILILC{} in context of the printhead simulations.
Furthermore, \cite{Chu2013} gives theorems stating that gradient ILC's convergence plateaus for linear \NMP{} systems, so it is possible that a similar issue arises for \PWA{} systems.

In addition to 
the preceding quantitative analysis,
inspection of Figure \ref{fig:timeseriesA3combined} reveals a qualitative comparison worth making between the ILC schemes: that of noise acquisition. While even qualitatively, \ILILC{} clearly demonstrates the highest quality reference tracking, it also accumulates the most noise in its learned feedforward input. There may be applications in which this noise acquisition is unattractive. However, it must also be noted that gradient ILC and P-type ILC are not free from unwanted frequency content. P-type ILC especially appears to develop mid-frequency oscillations that substantially degrade performance. Gradient ILC is more subtle; its main source of error appears to be oscillations arising from overshoot of the feedforward input.
Of course, learning-\emph{free} stable inversion 
(representing contributions \ref{contribution:inversion}-\ref{contribution:SI} of this article)
is completely free of noise accumulation because it is generated purely from a dynamic model.
Noise observable in the output under learning-free stable inversion comes exclusively from the sensitivity of the underlying system 
rather than the feedforward control itself.

In all, this study constitutes the first demonstration of ILC applied to any hybrid system with \NMP{} component dynamics. 
Stable inversion is a required part of the highest performing ILC scheme used, validating this article's theoretical contributions' utility for high performance control.

\section{Conclusion}
\label{sec:conc5}
This article has derived theory for the inversion of a class of \PWA{} systems.
This includes the implicit inverse formula for systems of any relative degree and explicit formulas for multiple cases in which they exist.
Specifically, proof of sufficient conditions for inverse explicitness, and the explicit formulas themselves, are given for systems with global dynamical relative degree (a concept introduced in this work) of 0, 1, or 2.
Additionally, for cases in which the inverse system is unstable, a stable inversion procedure is created, along with proof of the sufficient conditions for the procedure to be applicable.
The ability to analytically produce inverse system models for hybrid systems has multiple applications in controls. Demonstrated here is the newfound ability to apply ILC to \PWA{} systems with unstable inverses to achieve low error output reference tracking.

There are many avenues for future work
to overcome limitations in the proposed methods.
Of particular interest is the relaxation of the constraints on relative degree. For some hybrid systems, it may be desirable to have locations in which the input cannot affect the output, and the state is governed by natural dynamics alone. In such cases the global dynamical relative degree would be undefined (infinite), which is not considered here. There may also be more cases in which different locations feature different component relative degrees, which would violate \ref{C:muc}. Relaxing these constraints would dramatically expand the class of systems addressed. 
Extension to input-based switching 
(which may take inspiration from the likes of \cite{Hejri2021}) 
and multi-input-multi-output systems
would also be significant contributions.
A future contribution regarding stable inversion could be a method to treat non-hyperbolic systems (e.g. those in which (\ref{eq:decouplable}) has an eigenvalue equal to 1 $\forall k$), which are precluded by \ref{C:decouplable} in the present article.

Finally, while the present work makes preliminary testing and observation of noise robustness in the ILC of \PWA{} systems, analytical derivation of 
performance
bounds under trial-varying noise is an open problem. Future work in this space may take inspiration from prior art analyzing robustness to trial-invariant disturbances for the broad set of nonlinear systems under which \PWA{} systems fall \cite{Avrachenkov1998}, and robustness to noise for ILC of purely linear systems \cite{Hu2022,DeRoover1996,Kavli1992}.

\bibliographystyle{IEEEtran}
\bibliography{spiegelBibURL}

\appendix
\subsection{Example System Details \& Derivation}
\label{sec:A1}
System identification of the printer yields a discrete-time 
LTI
truth model.
Truncation-based model order reduction by 1 (MATLAB function \texttt{balred}), zero-order-hold-based sample period reduction by a factor of 2 (MATLAB function \texttt{d2d}), and random perturbation of model parameters yields a new LTI model, which is used 
as the control 
model and can be represented by the state space system ($\Acm^{P}$, $\Bcm^{P}$, $\Ccm^{P}$, $\Dcm^{P}$) with state vector $\xcm^{P}\kdx\in\real^{n_{\xcm^{P}}}$.
To account for the change in sample period, truth 
model
outputs
are decimated by a factor of 2 
before 
delivery
to the ILC law 
and ILC law 
outputs
are upsampled 
by 
2
(zero order hold)
before being applied to the truth model.
A Bode plot of the experimental data, truth model, and control model are given in Figure \ref{fig:bode}.

\begin{figure}
    \centering
    \includegraphics[scale=1]{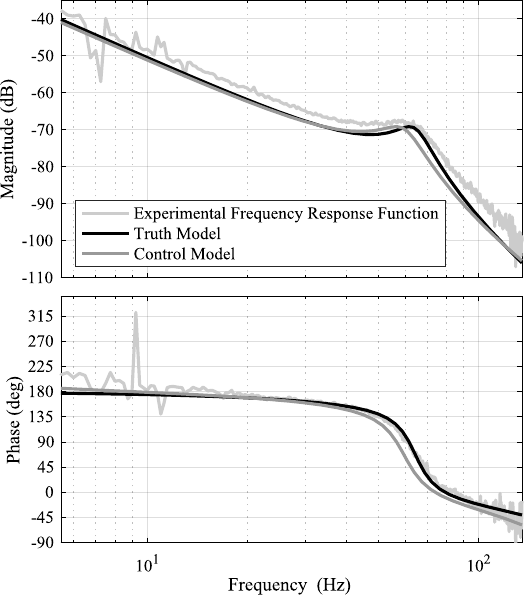}
    \caption{Bode plot of the experimental plant data, truth model of the plant, and control model of the plant.}
    \label{fig:bode}
\end{figure}

The feedback controller is of identical structure and tuning for the truth and control models, but has different parameter values due to the difference in sample period between the two models. 
The
feedback controller is composed of a second order lowpass filter given by the discrete-time transfer function
\begin{equation}
\label{eq:LP}
    C^{\textrm{LP}}(z) = \frac{b z(z+1)}{z^2+a_1z+a_2}
\end{equation}
in series with the 
switching
PD 
controller
{\allowdisplaybreaks\begin{IEEEeqnarray}{RL}
    C^{\textrm{PD}}(z)
    &= 
    \frac{
    \left(K_p + \frac{K_d}{T_s}\right)z - \frac{K_d}{T_s}
    }{z}
\end{IEEEeqnarray}}
Gain
$K_p$ is set to a high value when 
reference-output error 
exceeds a
magnitude threshold, and 
to a lower value otherwise.
Switching
is
based on
the error of the previous time step rather than the current time step in order for switching to be state-based, and thus satisfy \ref{C:stateswitch}. Mathematically this switching is made state-based by augmenting the minimal state-space representation of $C^{\textrm{PD}}(z)C^{\textrm{LP}}(z)$ with an extra state that stores the error input to the lowpass filter, yielding (\ref{eq:feedbackSS}).
For both the truth model and 
control model, the lowpass filter has a roll off frequency of 
\SI{40}{\hertz} 
and a damping 
ratio
of 0.7.

\subsection{ILC Filtering \& Tuning}
\label{sec:A2}
Adding filters to ILC is common practice (see, e.g. \cite{Bristow2006}). Two filters are used in this article. First, a zero-phase-shift version of the feedback controller's lowpass filter is applied to the input and output of the ILC law.
Filtering the ILC law output governs the weighting of frequency ranges on which learning is performed. Lowpass filtering reduces accumulation of content beyond the filter bandwidth (e.g. noise) into the feedforward input \cite{Bristow2006,DeRoover1996,Kavli1992}.
In lifted form, the feedback controller's filter is given by the lower diagonal, square, Toeplitz matrix $\mathscr{F}$ whose first column is the unit magnitude impulse response of the lowpass filter (\ref{eq:LP}) on $k\in\llbracket0,N-\dmu\rrbracket$. The zero phase shift is achieved by first filtering the signals forwards in time, and then backwards in time. The resultant lifted zero-phase-shift lowpass filter is
\begin{equation}
    \mathscr{Q} = \cancel{I}\mathscr{F}\cancel{I}\mathscr{F} 
\end{equation}
where $\cancel{I}$ is 
the exchange matrix (ones on the antidiagonal, zeros elsewhere).
Second, to eliminate time series edge effects the first 35 and last 35 samples of the 1000 sample ILC law output are forced to zero. These edge effects may arise because the finite stable inversion trajectories have nearly zero---rather than zero---initial conditions for the unstable modes and similar for the terminal conditions of the stable modes. In lifted form this filter is given by the identity matrix with the first 35 and last 35 diagonal elements set to zero, notated as $\mathscr{E}\in\real^{N-\dmu+1\times N-\dmu+1}$.

For gradient ILC and P-type ILC, tuning of the learning gains is done to achieve the most aggressive 
controller possible. Specifically $\gradILCgain$ and $\PILCgain$ are chosen as the largest whole number such that the NRMSE decreases monotonically over all trials. 
These numbers are found via a line search over many ILC 
simulations,
varying only the learning gains. 
By this method, $\gradILCgain=4255$ (dimensionless) and $\PILCgain=\SI{27}{\volt\per\meter}$. If the learning gains are increased above these values, the benchmark ILC schemes begin to exhibit instability.

\begin{IEEEbiography}[
{\includegraphics[width=1in,height=1.25in,clip,keepaspectratio]{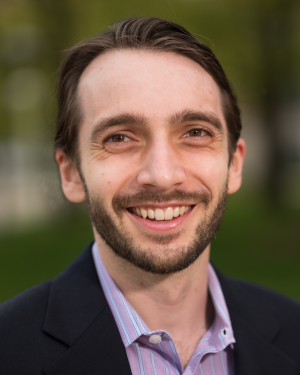}}
]{Isaac A. Spiegel}
received the B.S. degree 
(with honors)
in mechanical engineering and the 
Minor degree in Japanese 
from the University of California 
at Berkeley, 
Berkeley, CA, USA, in 2014. He received the 
M.S. and 
Ph.D. degrees in mechanical engineering in 
2019 and 
2021 from the University of Michigan, Ann Arbor, MI, 
USA, where he 
continues to contribute as a postdoctoral researcher. He is also a guidance navigation and control engineer at Terran Orbital.
His research focuses on developing hybrid modeling frameworks and optimization-based control theory that leverages those frameworks to address challenges encountered in practice with manufacturing, motion-control, and spacecraft systems.
\end{IEEEbiography}

\begin{IEEEbiography}[
{\includegraphics[width=1in,height=1.25in,clip,keepaspectratio]{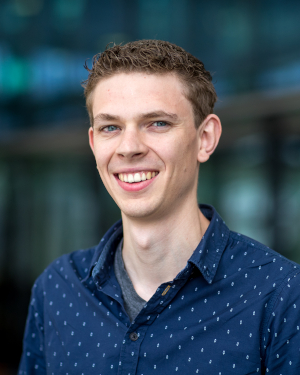}}
]{Nard Strijbosch}
received the B.Sc. (cum laude), M.Sc. (cum laude), and Ph.D. degrees from the Eindhoven University of Technology, Eindhoven, The Netherlands. He is a recipient of the IEEJ Industry Applications Society Excellent Presentation Award (SAMCON 2018). His research interest is in the field of motion control and learning control techniques for applications in mechatronic systems.
\end{IEEEbiography}

\begin{IEEEbiography}[
{\includegraphics[width=1in,height=1.25in,clip,keepaspectratio]{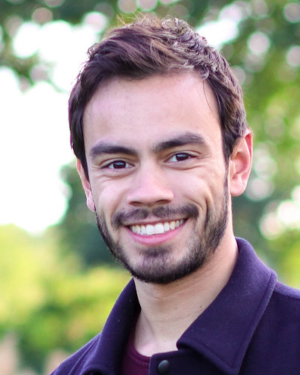}}
]{Robin de Rozario}
received the M.Sc. degree (cum laude) in mechanical engineering and the Ph.D. degree with the Control Systems Technology Group at the Eindhoven University of Technology, Eindhoven, the Netherlands, in 2015 and 2020, respectively. His research interests include system identification and learning control for high-performance motion systems.
\end{IEEEbiography}

\begin{IEEEbiography}[
{\includegraphics[width=1in,height=1.25in,clip,keepaspectratio]{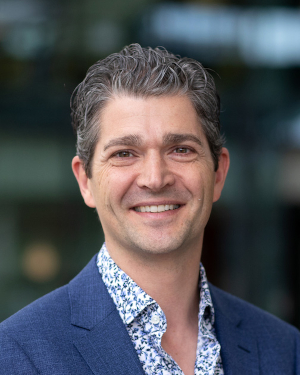}}
]{Tom Oomen}
received the M.Sc. (cum laude) and Ph.D. degrees from the Eindhoven University of Technology, Eindhoven, The Netherlands (TU/e). He is currently a full professor 
with the Department of Mechanical Engineering 
at TU/e and a part-time full professor with the Delft University of Technology. He held visiting positions at KTH, Stockholm, Sweden, and at The University of Newcastle, Australia. He is a recipient of the 7th Grand Nagamori Award, the Corus Young Talent Graduation Award, the IFAC 2019 TC 4.2 Mechatronics Young Research Award, the 2015 IEEE Transactions on Control Systems Technology (TCST) Outstanding Paper Award, the 2017 IFAC Mechatronics Best Paper Award, the 2019 IEEJ Journal of Industry Applications Best Paper Award, and Veni and Vidi personal grants. He is a Senior Editor of IEEE Control Systems Letters (L-CSS) and Associate Editor of IFAC Mechatronics, and he has served on the editorial boards of the IEEE Control Systems Letters (L-CSS) and IEEE TCST. He has also been vice-chair for IFAC TC 4.2 and a member of the Eindhoven Young Academy of Engineering. His research interests are in the field of data-driven modeling, learning, and control, with applications in precision mechatronics.
\end{IEEEbiography}

\begin{IEEEbiography}[
{\includegraphics[width=1in,height=1.25in,clip,keepaspectratio]{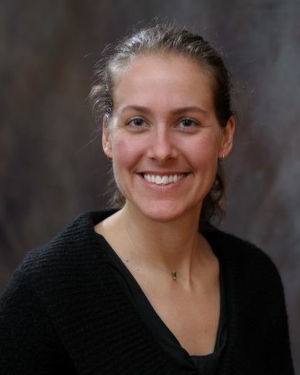}}
]{Kira Barton}
received the B.Sc. degree in mechanical engineering from the University of Colorado at Boulder (2001), and the M.Sc. and Ph.D. degrees in mechanical engineering from the University of Illinois at Urbana-Champaign (UIUC), in 2006 and 2010, respectively. She was a postdoctoral researcher at UIUC from Fall 2010 until Fall 2011, at which point she joined the Mechanical Engineering Department at the University of Michigan, Ann Arbor. She is currently an Associate Professor with the Department of Mechanical Engineering and a Core Robotics Faculty Member at the University of Michigan. She conducts research in modelling, sensing, and control for applications in advanced manufacturing and robotics, with a specialisation in cooperative learning and smart manufacturing. She was selected as a Miller Faculty Scholar from 2018-2021, and was a recipient of an NSF CAREER Award, in 2014, the 2015 SME Outstanding Young Manufacturing Engineer Award, the 2015 University of Illinois, Department of Mechanical Science and Engineering Outstanding Young Alumni Award, the 2016 University of Michigan, Department of Mechanical Engineering Department Achievement Award, and the 2017 ASME Dynamic Systems and Control Young Investigator Award.
\end{IEEEbiography}

\end{document}